\documentclass[12pt]{article}
\usepackage[utf8]{inputenc}
\usepackage[margin=1in]{geometry}
\usepackage{natbib}
\usepackage{robinstyle}
\usepackage{hyperref}
\usepackage[T1]{fontenc}
\usepackage[ruled, vlined]{algorithm2e}
\usepackage{multirow, subcaption, chemformula, authblk}
\usepackage{dsfont}

\makeatletter
\SetKwInOut{Input}{Input}
\SetKwInOut{Output}{Output}
\SetKwInOut{Initialize}{Initialize}

\newcommand{\robin}[1]{#1}
\newcommand{\lammax}{\lambda_{\mathrm{max}}}
\newcommand{\lammin}{\lambda_{\mathrm{min}}}
\DeclareMathOperator{\diag}{diag}
\newcommand{\mU}{\mat U}
\newcommand{\mW}{\mat W}
\newcommand{\mUW}{[\mU, \mW]}
\newcommand{\SUW}{\mat\Sigma_{\mU\mW}}
\newcommand{\hSUW}{\hvec\Sigma_{\mU\mW}}
\newcommand{\snr}{\mathrm{SNR}}

\title{Convex Estimation of Gaussian Graphical Regression Models with Covariates}
\author{Ruobin Liu}
\author{Guo Yu}
\affil{Department of Statistics and Applied Probability, University of California, Santa Barbara, Santa Barbara, CA, USA}
\date{}
\begin{document}
\maketitle

\begin{abstract}%
Gaussian graphical models (GGMs) are widely used to recover the conditional independence structure among random variables. Recent work has sought to incorporate auxiliary covariates to improve estimation, particularly in applications such as co-expression quantitative trait locus (eQTL) studies, where both gene expression levels and their conditional dependence structure may be influenced by genetic variants.
Existing approaches to covariate-adjusted GGMs either restrict covariate effects to the mean structure or lead to nonconvex formulations when jointly estimating the mean and precision matrix. In this paper, we propose a convex framework that simultaneously estimates the covariate-adjusted mean and precision matrix via a natural parametrization of the multivariate Gaussian likelihood.
The resulting formulation enables joint convex optimization and yields improved theoretical guarantees under high-dimensional scaling, where the sparsity and dimension of covariates grow with the sample size. We support our theoretical findings with numerical simulations and demonstrate the practical utility of the proposed method through a reanalysis of an eQTL study of glioblastoma multiforme \robin{and an analysis of diet on the human gut microbiome}.
\end{abstract}

\section{Introduction}
Graphical models use graphs to represent conditional independence relationships among components of a random vector $\bm X=(X_1,\dots,X_p)$.
In the Gaussian case, this leads to Gaussian graphical models (GGMs), where conditional independences correspond to zeros in the precision matrix $\mat\Omega=\mat\Sigma^{-1}$.
Specifically, for a Gaussian random vector $\bm X$ and $i\neq j$, the variables $X_i$ and $X_j$ are conditionally independent given $\bm X_{\{1,\dots,p\}\setminus\{i,j\}}$ if and only if $\Omega_{ij}=0$ \citep{lauritzen}.
There has been extensive work on the estimation of GGMs in high-dimensional settings \citep{yuanlin, friedman-glasso, meinshausen}, where sparsity is typically induced on the precision matrix to improve interpretability and estimation accuracy.

A major application of GGMs arises in genomics, where the goal is to infer gene interaction networks from gene expression data.
Such problems are inherently high-dimensional, with the number of genes often exceeding the number of samples \citep{schafer2005}.
A prominent example is provided by co-expression quantitative trait locus (eQTL) studies, in which gene expression measurements are analyzed jointly with external genetic markers.
These markers may confound gene expression levels, either through local genetic effects or through associations with unmeasured factors influencing gene expression.
This setting motivates the development of Gaussian graphical models that incorporate covariate information, enabling adjustment for such confounding in network estimation.

Specifically, let $\bm X \in \R^p$ be the vector of responses and $\bm U\in\R^q$ the associated vector of covariates.
We consider the general model
\begin{equation}
\label{eq:multitask}
\bm X\given \bm U = \vec u \sim N(\vec\mu(\vec u), \mat\Omega^{-1}(\vec u)),
\end{equation}
which allows both the mean vector and the precision matrix of $\bm X$ to depend on a realization of covariates $\bm U = \vec u$.
A commonly used specification assumes that only the mean depends on $\bm U$, while the precision matrix remains constant. This specification is also known as the conditional Gaussian graphical model \citep{yin-li}:
\begin{equation}
\label{eq:meanadj}
\vec\mu(\vec u) = \mat\Gamma \vec u,\quad
\mat\Omega(\vec u) = \mat B_0.
\end{equation}
Estimation of the unknown parameters in \eqref{eq:meanadj} is challenging, as the negative log-likelihood associated with \eqref{eq:multitask} is not jointly convex in $(\mat\Gamma, \mat B_0)$.
Existing approaches typically rely on alternating optimization \citep{rothman,yin-li,chen18} or two-stage procedures \citep{cai13, chen2016}.
Despite this nonconvexity, minimax-optimal rates for joint estimation under \eqref{eq:multitask} with \eqref{eq:meanadj} have been established \citep{chen18, lv}.
Several works have sought to convexify the problem. For instance, \citet{wang_joint_2015} propose a convex formulation based on an initial estimate of $\mat\Gamma$, while \citet{zhu_convex_2020} introduce a reparametrization that yields a jointly convex negative log-likelihood.

Model \eqref{eq:multitask} with specification \eqref{eq:meanadj} can also be viewed as a multivariate response regression model, where the primary goal is to improve estimation of $\mat\Gamma$ by leveraging structural assumptions on $\mat B_0$.
In contrast, the covariate-adjusted GGM framework focuses on recovering the sparsity pattern of the precision matrix while accounting for covariate effects through $\bm U$.
Several methods in this setting incorporate structure-inducing penalties into the likelihood of \eqref{eq:multitask} associated with the specification \eqref{eq:meanadj}; see \citet{liu_estimation_2025} for a recent review highlighting connections and distinctions between these closely related lines of work.

While substantial progress has been made in modeling heterogeneous means, relatively little work allows the graph structure itself to vary with covariates; most existing approaches instead assume a common precision matrix $\mat B_0$ across subjects.
However, subject-specific graph structures are plausible in applications such as eQTL studies.
For example, gene C may mediate the co-expression of genes A and B only in the presence of a single-nucleotide polymorphism (SNP) local to gene C.
In such cases, genes A and B may be conditionally independent given the rest of the network unless a genetic variant is present near gene C \citep{fehrmann,kolberg,rockman}.

A small number of works have considered models in which both the mean and precision depend on covariates.
Bayesian approaches, for example, allow subject-specific precision matrices by imposing regression or clustering structures on $\mat\Omega$ \citep{ni_bayesian_2022, niu_covariate-assisted_2024, wang_bayesian_2022, zeng_bayesian_2025}.
However, these methods do not directly address joint estimation of both the mean and precision as functions of covariates within a unified framework.

This work is closely related to \citet{zhangli}, which extends \eqref{eq:meanadj} to allow both the mean and the network structure of the responses to depend on covariates.
Specifically, their model specifies \eqref{eq:multitask} by taking
\begin{equation}
\label{eq:zhanglispec}
\vec\mu(\vec u) = \mat\Gamma \vec u,\quad
\mat\Omega(\vec u) = \mat B_0 + \sum_{h=1}^q \mat B_h u_h
\end{equation}
where $\mat B_h \in \R^{p \times p}, h = 0, 1, \dotsc, q$ are sparse, symmetric matrices.
As in conditional GGMs, joint estimation of the parameters $(\mat\Gamma, \mat B_0, \mat B_1, \dotsc, \mat B_q)$ under \eqref{eq:zhanglispec} leads to a nonconvex objective. To address this, \citet{zhangli} adopt a two-stage estimation procedure.
Building on this framework, \citet{zhang-multi-task-2025} propose a multi-task estimator with a cross-task group penalty, while \citet{meng_statistical_2025} develop debiased nodewise estimators for inference on edge--covariate effects.
\robin{
\citet{wang_high-dimensional_2025} introduce a parametrization for covariate-adjusted graphs where the estimated precision matrices are guaranteed to be positive definite. However, their work is developed for small, fixed numbers of covariates $q$ and their theoretical guarantees do not immediately allow for $q$ to scale with $n$.
}

In this work, we develop \emph{Covariate-adjusted Sparse Precision with Natural Estimation} (\texttt{cspine}), a method for estimating covariate-adjusted GGMs with heterogeneous network effects based on a jointly convex formulation of \eqref{eq:multitask}.
Like \citet{zhangli}, our method uses nodewise regression \citep{meinshausen} to estimate the covariate-adjusted graph structure.
However, by adopting a natural parametrization of the Gaussian likelihood, we obtain a formulation in which each nodewise regression reduces to a convex optimization problem, thereby avoiding the nonconvexity inherent in the existing approaches.

Our main contributions and the organization of the paper are as follows.
In Section~\ref{sec:setup}, we review the framework of \citet{zhangli} and introduce the proposed natural parametrization, under which the likelihood in \eqref{eq:multitask} is jointly convex in all unknown parameters.
\robin{Moreover, the natural parametrization also carries an interpretable mean structure whereby the average level of the responses is mediated by the covariance matrix. This property is detailed in Section~\ref{sec:nat_effect}}.
Building on this formulation, Section~\ref{sec:penalized} presents the proposed model and estimation procedure with sparsity-inducing penalties.
Section~\ref{sec:theory} establishes theoretical guarantees for the proposed estimator, showing that the convex formulation leads to improved scaling of the ambient dimensions $p$ and $q$ relative to the sample size $n$ under assumptions comparable to those in \citet{zhangli}.
These theoretical advantages are further supported by extensive simulations in Section~\ref{sec:simulation}. Finally, in Section~\ref{sec:realdata}, we illustrate the practical utility of the proposed approach through an analysis of glioblastoma multiforme (GBM) gene expression data \robin{and human gut microbiome data}.
Proofs of all theoretical results are provided in Appendix~\ref{appdx:proofs}.

Throughout this paper, we adopt the following notational conventions unless otherwise specified. Scalars are denoted by regular (non-bold) letters, vectors by bold lowercase letters, and matrices by bold uppercase letters.
For any positive integer $m$, let $[m]=\{1, \ldots, m\}$, and for $j \in [m]$, let $-j$ denote the index set $[m]\setminus\{j\}$. For a matrix $\mat M$ and index sets $\cI$ and $\cJ$, we write $\mat M_{\cI,}$, $\mat M_{\cJ}$, and $\mat M_{\cI,\cJ}$ for the submatrices of $\mat M$ consisting of rows indexed by $\cI$, columns indexed by $\cJ$, and both rows indexed by $\cI$ and columns indexed by $\cJ$, respectively.

\section{Model Formulation and Natural Parametrization}
\label{sec:setup}
In this section, we first review the model formulation of \citet{zhangli} for estimating covariate-adjusted GGMs in \eqref{eq:multitask} via a nodewise regression approach, and discuss the challenges of their two-stage estimation procedure arising from the nonconvexity of the formulation.
We then introduce a convex formulation of \eqref{eq:multitask} based on a natural parametrization.

\subsection{Nodewise Regression in Covariate-Adjusted GGMs}
Suppressing the dependence of $\mat\Sigma$ and $\vec\mu$ on the covariates $\bm U$ for notational simplicity, the Gaussian model in \eqref{eq:multitask} implies that, for each $j \in [p]$,
\begin{equation}
\label{eq:conditional}
X_j - \mu_j \given \bm X_{-j}, \bm U
\sim N\parens*{\mat\Sigma_{j, -j}\mat\Sigma^{-1}_{-j,-j}(\bm X_{-j} - \vec\mu_{-j}),\, \Sigma_{jj} - \mat\Sigma_{j, -j}\mat\Sigma^{-1}_{-j,-j}\mat\Sigma_{-j, j}}.
\end{equation}
By the block matrix inversion formula, \eqref{eq:conditional} is equivalent to the linear model with Gaussian errors
\begin{equation}
\label{eq:nodewise}
X_j = \mu_j - \Omega^{-1}_{jj}\mat\Omega_{j,-j}(\bm X_{-j} - \vec\mu_{-j}) + \ep_j,
\quad \ep_j \sim N(0, \sigma^2_{\ep_j}), \quad \sigma^2_{\ep_j} = 1/\Omega_{jj},
\end{equation}
where the expression is understood conditional on $(\bm X_{-j}, \bm U)$.
Based on this representation, nodewise regression methods for estimating GGMs \citep{meinshausen} can be naturally extended to incorporate covariate adjustments, where the regression coefficients and error variance are used to recover the $j$-th column of the precision matrix $\mat\Omega$.

\subsection{Non-Convexity of the Nodewise Regression Under \eqref{eq:zhanglispec}}
\label{sec:zhangli}
Write $\mat\Gamma = (\vec\gamma_1,\dotsc, \vec\gamma_p)^\T \in \R^{p \times q}$. Under the specification \eqref{eq:zhanglispec} of \citet{zhangli}, the nodewise regression model \eqref{eq:nodewise} becomes
\begin{equation}
\label{eq:zhanglinodewise}
X_j = \vec u^\T \vec\gamma_j + \sum_{k\neq j}^p \eta_{jk0}(X_k - \vec u^\T \vec\gamma_k) + \sum_{k\neq j}^p \sum_{h=1}^q \eta_{jkh} u_h(X_k - \vec u^\T \vec\gamma_k) + \ep_j,
\end{equation}
where $\eta_{jkh} = -[\mat B_h]_{jk}/\Omega_{jj}$ for $k \in [p] \backslash \{j\}$ and $h \in \{0 \} \cup [q]$; this model is termed \emph{Gaussian graphical regression} in \citet{zhangli}.
Indeed, the least squares criterion based on \eqref{eq:zhanglinodewise} is not jointly convex in the unknown parameters, due to the product terms $\eta_{jkh}\times\vec u^\T \vec\gamma_k$.

To address this difficulty, \citet{zhangli} separate the estimation of $\vec\gamma_j$ from that of $\eta_{jkh}$.
Specifically, if $\mat\Gamma$ were known, the responses may first be centered via $Z_j = X_j - \vec u^\T \vec \gamma_j$, reducing the model to the oracle form
\begin{equation}
\label{eq:zhanglinodewise_oracle}
Z_j = \sum_{k\neq j}^p \eta_{jk0} Z_k + \sum_{k\neq j}^p \sum_{h=1}^q \eta_{jkh} u_h Z_k + \ep_j,
\end{equation}
for which the least squares criterion is convex in the $\eta$'s.

In practice, however, $\mat\Gamma$ is not known, and a two-stage procedure is employed.
In the first stage, each $\vec\gamma_j$ is estimated under a simplified working model
\[
X_j = \vec u^\T\vec\gamma_j + \ep_j.
\]
Comparing with the full model \eqref{eq:zhanglinodewise}, this working model imposes a strong assumption on the specification \eqref{eq:zhanglispec}, effectively treating the precision matrix $\mat B_0$ as diagonal and all other matrices $\mat B_h$ as zero matrices, thereby assuming conditional independence among all pairs of responses, which is furthermore not dependent on covariates.
As a result, the first-stage estimator does not account for the network structure and its potential dependence on covariates, and thus may incur non-negligible bias when such dependencies are present.

Given the resulting estimate of $\mat\Gamma$, the second stage proceeds by forming centered responses $Z_j$ and estimating the coefficients $\eta_{jkh}$ using the oracle model \eqref{eq:zhanglinodewise_oracle}.
This stage treats the estimate of $\mat\Gamma$ as fixed, thereby inheriting any estimation error from the first stage.

The two-stage procedure introduces additional estimation error, as inaccuracies in estimating $\vec\gamma_j$ propagate to the estimation of $\eta_{jkh}$.
As shown in the theoretical analysis of \citet{zhangli}, the concentration guarantees are weakened relative to the oracle setting in which $\mat\Gamma$ is known, since the overall error probability depends on the tuning parameters from the first stage. We revisit these limitations in Section~\ref{sec:theoretical_discussion}.

\subsection{A Convex Formulation via Natural Parametrization}
These limitations motivate the development of a formulation that enables joint convex estimation without relying on stagewise procedures. Specifically, we reparametrize the Gaussian likelihood using its natural parameters, which decouple the interaction between regression and precision parameters in \eqref{eq:zhanglinodewise}.
Recall that the density of a $p$-dimensional Gaussian distribution can be expressed in terms of its natural parameters $(\tvec\theta, \tvec\Theta)$, where $\tvec\theta = \mat\Omega\vec\mu$ and $\tvec\Theta = -\mat\Omega$, as
\begin{equation}
\label{eq:natural_gauss}
L(\tvec\theta, \tvec\Theta \given \vec x) =
\exp\sets*{\tvec\theta^\T \vec x - \frac{1}{2}\vec x^\T\tvec\Theta \vec x - A(\tvec\theta, \tvec\Theta)}.
\end{equation}
Our formulation is motivated by the fact that the cumulant function $A$ is jointly convex in the natural parameter $(\tvec\theta, \tvec\Theta)$ and therefore so is $-\log L(\tvec\theta, \tvec\Theta \given \vec x)$.
Specifically, define
\begin{equation}
\label{eq:reparam}
\vec\theta = \diag(\mat\Omega)^{-1}\mat\Omega\vec\mu
\quad\text{and}\quad
\mat\Theta = -\diag(\mat\Omega)^{-1}\mat\Omega,
\end{equation}
where $\diag(\mat\Omega)$ is the $p\times p$ diagonal matrix formed from $\mat\Omega$.
Under this parametrization, the nodewise regression model in \eqref{eq:nodewise} can be written as
\[
X_j = \mu_j - \Omega^{-1}_{jj}\mat\Omega_{j,-j}(\bm X_{-j} - \vec \mu_{-j}) + \ep_j
= \theta_j + \mat\Theta_{j,-j}\bm X_{-j} + \ep_j,
\]
which yields a least squares criterion that is jointly convex in $\theta_j$ and $\mat\Theta_{j,-j}$.

In contrast to \eqref{eq:zhanglispec}, we model covariate dependence through $(\vec\theta, \mat\Theta)$ rather than $(\vec \mu, \mat\Omega)$ by specifying
\begin{equation}
\label{eq:cspinespec}
\vec\theta(\vec u) = \mat\Gamma \vec u,\quad \mat\Theta(\vec u) = \mat B_0 + \sum_{h=1}^q \mat B_hu_h,
\end{equation}
where $\mat B_0 \in \R^{p \times p}$ represents the baseline graph structure and $\mat B_h \in \R^{p \times p}$ captures the effect of covariate $u_h$ on the graph structure.
Writing $\mat\Gamma = (\vec\gamma_1,\dotsc, \vec\gamma_p)^\T$ and letting $\beta_{jkh} = [\mat B_h]_{jk}$ for $k \in [p]\setminus\{j\}$ and $h \in \{0\} \cup [q]$,
this specification leads to the following nodewise regression model
\begin{align}
X_j = \vec u^\T\vec\gamma_j + \sum_{k\neq j}^p \beta_{jk0} X_k + \sum_{k\neq j}^p \sum_{h=1}^q \beta_{jkh} u_h X_k + \ep_j,
\quad \ep_j \sim N(0, \sigma^2_{\ep_j}), \quad \sigma^2_{\ep_j} = 1/\Omega_{jj}
\label{eq:nodewise1}
\end{align}
for $j\in[p]$.
That is, $X_j$ is regressed on the covariates $\vec u$, all other responses $X_k$ for $k \neq j$, and the pairwise interactions between the covariates and the other responses.
Remarkably, the least squares criterion associated with \eqref{eq:nodewise1} is jointly convex in $\vec\gamma_j$ and $\beta_{jkh}$, enabling simultaneous estimation of all parameters and avoiding the nonconvexity inherent in the formulation of \citet{zhangli}.

The parametrization \eqref{eq:reparam} implicitly imposes that $\Theta_{jj} = -1$ for all $j\in[p]$.
Under the specification \eqref{eq:cspinespec}, this is enforced by setting $[\mat B_0]_{jj} = -1$ and $[\mat B_h]_{jj} = 0$ for all $h\in[q]$ and $j\in[p]$.
Consequently, the residual variance $\sigma^2_{\ep_j} = 1/\Omega_{jj}$, that is, the conditional variance, does not depend on covariates, an assumption also adopted by \citet{zhangli}.
Moreover, $\mat\Theta$ and $\mat\Omega$ share the same sparsity structure, so that estimating the sparsity pattern of $\mat\Theta$ is equivalent to that of $\mat\Omega$.
Accordingly, we focus on recovering the sparsity pattern of $\mat\Theta$ in what follows.

Finally, we note that the proposed parametrization is closely related to, yet distinct from, existing convex formulations.
The transformation in \eqref{eq:reparam} is conceptually related to the natural parametrization in \citet{yu2019} for univariate-response regression and its multivariate extension in \citet{zhu_convex_2020}.
However, the scaling by $\diag(\mat\Omega)^{-1}$ in \eqref{eq:reparam} is crucial for ensuring that the nodewise regression model in \eqref{eq:nodewise1} remains jointly convex in the presence of covariate adjustments.
This feature distinguishes our approach from the reparametrizations in \citet{yu2019} and \citet{zhu_convex_2020}, which do not accommodate covariate-dependent variance or covariance structures.

\subsection{Effect of Covariates and Network Structure on the Mean}
\label{sec:nat_effect}
Let $\tvec\theta(\vec u) = \mat\Omega(\vec u)\vec\mu(\vec u)$ denote the natural parameter as in \eqref{eq:natural_gauss}.
Combining with \eqref{eq:reparam} and \eqref{eq:cspinespec} gives the mean as $\vec\mu(\vec u) = \mat\Sigma(\vec u) \tvec\theta(\vec u)$, where $\mat\Sigma(\vec u) = \mat\Omega^{-1}(\vec u)$ is the covariance of $\bm X$ given $\bm U = \vec u$. The matrix $\mat\Sigma(\vec u)$ acts as a propagation operator for the natural parameter $\tvec\theta(\vec u)$ so that the product aggregates the contributions of $\vec u$ across co-varying responses to produce $\vec\mu(\vec u)$.
Explicitly, the conditional mean of $X_j$ is decomposed as
\begin{equation}
\label{eq:cond_mean_decomp}
\mu_j(\vec u) = \Sigma_{jj}(\vec u)\,\Omega_{jj}\vec\gamma_j^\T \vec u + \sum_{k \neq j} \Sigma_{jk}(\vec u)\,\Omega_{kk}\vec\gamma_k^\T \vec u.
\end{equation}
The first term in \eqref{eq:cond_mean_decomp} captures the contribution of $\vec u$ to $\mu_j(\vec u)$ through node $j$ itself.
The product $\Sigma_{jj}(\vec u)\,\Omega_{jj}$ is the variance inflation factor $1/(1-R^2_j(\vec u))$, where $R^2_j(\vec u)$ is the $R^2$ of regressing $X_j$ on $\bm X_{-j}$ conditional on $\bm U = \vec u$.
The intercept $\vec\gamma_j^\T \vec u$ in \eqref{eq:nodewise1} is therefore amplified by a factor that grows with how predictable $X_j$ is from its network neighbors.
The second term in \eqref{eq:cond_mean_decomp} aggregates contributions from all other nodes $k \neq j$.

The decomposition \eqref{eq:cond_mean_decomp} allows covariates to have cascading effects on the mean through the estimated network, a plausible phenomenon in many domains.
For instance, recent advances in the understanding of gene regulatory networks point toward network-mediated effects of eQTLs on mean expression levels.
Empirical evidence demonstrates that many eQTLs may alter the expression of distant genes by percolating local perturbations through shared pathways and co-expression networks \citep{boyle_expanded_2017,wolfgang_2018,vosa_large-scale_2021,zhou_joint_2021,ruzickova}.
In the context of the human gut microbiome, the decomposition in \eqref{eq:cond_mean_decomp} allows host factors such as diet to influence the abundance of a given taxonomic unit both directly and indirectly, as perturbations to keystone species propagate through cross-feeding and competitive exclusion relationships to reshape the overall microbial community \citep{culp_cross-feeding_2023}.

\section{Estimation}
\label{sec:penalized}
For each $j \in [p]$, define the concatenated parameter vector
$\vec\beta_j = (\vec\beta^\T_{j,0},\vec\beta^\T_{j,1}, \dotsc, \vec\beta_{j,q}^\T)^\T \in \R^{(p-1)(q+1)}$,
where the $h$-th block collects the parameters
\[\vec\beta_{j,h} = (\beta_{j1h}, \dotsc, \beta_{j,j-1,h}, \beta_{j,j+1, h}, \dotsc, \beta_{jph})^\T \in \R^{p-1}.\]
Under this notation, the nodewise regression model \eqref{eq:nodewise1} can be written compactly as
\begin{equation}
\label{eq:nodewise2}
X_j = \vec u^\T\vec\gamma_j + \bm X_{-j}^\T \vec\beta_{j,0} + \sum_{h=1}^q (u_h \bm X_{-j})^\T\vec\beta_{j,h} + \ep_j,
\quad \ep_j \sim N(0, \sigma^2_{\ep_j}).
\end{equation}
Figure~\ref{fig:thetadecomp} illustrates how each block $\vec\beta_{j,h}^\T$ corresponds to the $j$-th row (and, by symmetry, the $j$-th column) of $\mat\Theta$ in \eqref{eq:cspinespec}.

Suppose we collect $n$ independent observations $\sets{(\vec x^{(i)}, \vec u^{(i)})}_{i=1}^n$ of responses and covariates from the joint distribution in \eqref{eq:multitask}.
Let $\mat U = \bracks*{\vec u^{(1)}, \dotsc, \vec u^{(n)}}^\T \in \R^{n\times q}$
and $\mat X = \bracks*{\vec x^{(1)}, \dotsc, \vec x^{(n)}}^\T \in \R^{n \times p}$ denote the covariate and response matrices, respectively. Let $\vec u_h \in \R^n$ and $\vec x_k \in \R^n$ denote the $h$-th and $k$-th columns of $\mat U$ and $\mat X$.
For $j\in[p]$, define the design matrix
$
\mat W_{-j} = \bracks*{\mat W_{-j,0},\ \mat W_{-j, 1},\ \dotsc,\ \mat W_{-j,q}} \in \R^{n\times (p-1)(q+1)},
$
where
\begin{align*}
\begin{split}
\mat W_{-j,0} &=\bracks*{\vec x_1, \dotsc, \vec x_{j-1}, \vec x_{j+1}, \dotsc, \vec x_p} \in \R^{n\times(p-1)},\\
\mat W_{-j, h} &=
\bracks*{\vec x_1 \odot \vec u_h,\ \dotsc,\ \vec x_{j-1}\odot \vec u_h,\ \vec x_{j+1} \odot \vec u_h,\ \dotsc,\ \vec x_{p}\odot \vec u_h}  \in \R^{n\times(p-1)}\; \text{for $h \in [q]$},
\end{split}
\end{align*}
and $\odot$ denotes the elementwise product.

We consider estimating $\vec\gamma_j$ and $\vec\beta_j$ in \eqref{eq:nodewise2} by solving the following penalized least squares problem
\begin{equation}
\label{eq:cspineobjgeneral}
\minimize_{\vec \gamma_j,\, \vec\beta_j}\,
\frac{1}{2n}\norm*{\vec x_j - \mat U\vec \gamma_j - \mat W_{-j}\vec\beta_j}_2^2 + g_j(\vec\gamma_j, \vec\beta_j),
\end{equation}
where $g_j$ is a convex penalty function.
In this paper, we focus on penalties that induce sparsity in $\vec\beta_j$, promoting interpretable graph estimates.
By construction, sparsity in $\vec\beta_j$ corresponds to sparsity in the $j$-th row (and, by symmetry, the $j$-th column) of $\mat\Theta$, and hence in $\mat\Omega$ via \eqref{eq:reparam}.

\begin{figure}[tbp]
\centering
\includegraphics[width=0.9\linewidth]{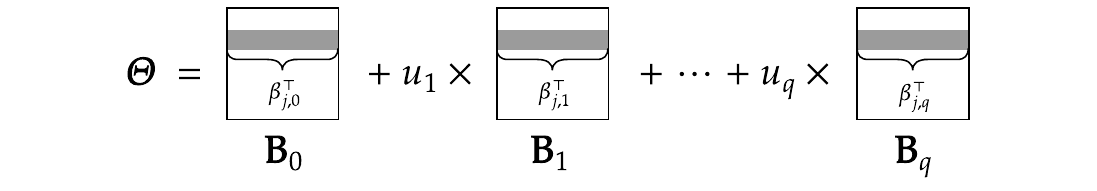}

\caption{Decomposition of $\mat\Theta$ into components $\mat B_h$ according to \eqref{eq:cspinespec}. The block $\vec\beta_{j, h}$ corresponds to the effects of covariate $u_h$ on the partial correlations of response $X_j$ while $\vec\beta_{j,0}$ describes the population effect.}
\label{fig:thetadecomp}
\end{figure}

\subsection{Sparse-Group Penalized Estimation}
Motivated by applications to eQTL studies, we consider a structured sparsity-inducing penalty for $g_j$ in \eqref{eq:cspineobj}.
Our goal is to capture both elementwise and groupwise sparsity in the coefficients $\vec\beta_j$.
Elementwise sparsity within each group $\vec\beta_{j,h}$ allows for the possibility that a covariate $u_h$ affects the conditional dependence between $X_j$ and only a subset of the remaining $\{X_k : k \neq j\}$.
In contrast, groupwise sparsity, where $\vec\beta_{j,h} = \vec 0$, indicates that the covariate $u_h$ has no effect on any conditional dependence structure involving $X_j$.

To this end, we consider the convex optimization problem
\begin{equation}
\label{eq:cspineobj}
\minimize_{\vec \gamma_j,\, \vec\beta_j}\,
\frac{1}{2n}\norm*{\vec x_j - \mat U\vec \gamma_j - \mat W_{-j}\vec\beta_j}_2^2 +  \lambda_j\norm{\vec\gamma_j}_1 + \lambda_j\norm{\vec\beta_j}_1 + \lambda_{g,j} \norm{\vec\beta_{j,-0}}_{1,2}
\end{equation}
where $\lambda_j \geq 0$ and $\lambda_{g,j} \geq 0$ are tuning parameters, and
$\norm{\vec\beta_{j,-0}}_{1,2} = \sum_{h=1}^q \norm{\vec\beta_{j,h}}_2$
denotes the group lasso penalty \citep{yuan2006model}.
\robin{To ease notation and with the understanding that the nodewise regressions \eqref{eq:cspineobj} are independent, we will suppress the notational dependence of $\lambda$ and $\lambda_g$ on the node $j$.}
Together, the $\ell_1$ and group penalties on $\vec\beta_j$ form a sparse-group lasso penalty \citep{simon}.
In contrast, only an $\ell_1$ penalty is applied to $\vec\beta_{j,0}$ and $\vec\gamma_j$, allowing for unrestricted elementwise sparsity in the baseline precision matrix $\mat B_0$ and in the covariate effects on the mean vector.

To estimate the diagonal entries of $\mat\Omega$, an estimate of the noise variance $\sigma^2_{\ep_j}$ is required. While several estimators are available in the literature on error variance estimation in high-dimensional linear regression \citep{reid,yu2019}, we adopt
\begin{align}
\hat\sigma^2_{\ep_j} =
\frac{\norm{\vec x_j - \mat U\hat{\vec\gamma}_j - \mat W_{-j} \hat{\vec \beta}_j}_2^2}
{n - \widehat{\mathrm{nnz}}},
\label{eq:noisevar}
\end{align}
where \robin{$\widehat{\mathrm{nnz}}$ denotes the number of nonzero entries of $(\hvec\gamma_j^\top, \hvec\beta_j^\top)^\top$} as in \citet{zhang_efficient_2020}.

\subsection{Optimization}
\label{sec:opt}
\robin{
Following \citet{zhang-multi-task-2025}, we use the augmented Lagrangian method with semismooth Newton subproblem solvers of \citet{zhang_efficient_2020} to solve \eqref{eq:cspineobj}. This method achieves superlinear local convergence by exploiting the second-order sparsity of the sparse-group proximal operator and we have found it to outperform first order methods in our experimentation.
Let $\mat A_j = [\mat U,\,\mat W_{-j}]$, $\vec c_j = \vec x_j$, and $\vec z_j = (\vec\gamma_j^\top,\vec\beta_j^\top)^\top$, so \eqref{eq:cspineobj} can be written as
$\min_{\vec z_j}\norm{\mat A_j\vec z_j - \vec c_j}_2^2/2n + g(\vec z_j)$,
where 
$g(\cdot)$ is the sparse-group lasso penalty.
Its dual problem is
\[
\min_{\vec v,\,\vec s}\;\langle\vec c_j,\vec v\rangle + \frac{n}{2}\|\vec v\|_2^2 + g^*(\vec s),
\qquad\text{s.t.}\quad \mat A_j^\T\vec v + \vec s = \vec 0,
\]
where $g^*$ is the convex conjugate of $g$ and $\vec s$ and $\vec v$ are dual variables.
This dual problem is solved via the augmented Lagrangian method (ALM), with the subproblem in each ALM update solved using a semismooth Newton step, which is outlined inAlgorithm~\ref{alg:ssnal}.}

\robin{Specifically, consider the proximal operator $\operatorname{prox}_g(\vec w) = \argmin_{\vec z}\{g(\vec z)+\norm{\vec z-\vec w}_2^2/2\}$ and set $\vec w^k(\vec v) = \tau_k^{-1}\vec z_j^k - \mat A_j^\T\vec v$.
The ALM subproblem at iteration $k$ amounts to minimizing
\begin{equation}
\label{eq:al_dual}
\psi_k(\vec v) = \langle\vec c_j,\vec v\rangle + \frac{n}{2}\|\vec v\|_2^2
  + g^*\!\bigl(\operatorname{prox}_{g^*/\tau_k}(\vec w^k(\vec v))\bigr)
+ \frac{\tau_k}{2}\bigl\|\operatorname{prox}_g(\vec w^k(\vec v))\bigr\|_2^2
  - \frac{1}{2\tau_k}\|\vec z_j^k\|_2^2.
\end{equation}
Step~1 in Algorithm~\ref{alg:ssnal} solves $\min_{\vec v}\psi_k(\vec v)$ via semismooth Newton iterations using the generalized Jacobian of $\operatorname{prox}_g$.
The proximal map decomposes as $\operatorname{prox}_g = \operatorname{prox}_{\lambda_g\|\cdot\|_{1,2}} \circ \operatorname{prox}_{\lambda\|\cdot\|_1}$, and only groups $G_h$ satisfying $\|[\operatorname{prox}_{\lambda\|\cdot\|_1}(\vec w)]_{G_h}\|_2 > \lambda_g$ contribute to the Jacobian.
This second-order sparsity means each Newton solve operates on a dense system whose size is determined by the solution's active set. We refer the reader to \citet{zhang_efficient_2020} for details.
}
\begin{algorithm}
\SetAlgoLined
Let $\tau_0>0$ and $(\vec v^0,\vec s^0,\vec z_j^0)$ be given. Iterate for $k=0,1,\dotsc$ until convergence:

\textbf{Step 1:} Compute $\vec v^{k+1}=\min_{\vec v}\,\psi_k(\vec v)$ \eqref{eq:al_dual} and $\vec s^{k+1}=\operatorname{prox}_{g^*/\tau_k}\!\left(\tau_k^{-1}\vec z_j^k-\mat A_j^\T\vec v^{k+1}\right)$

\textbf{Step 2:} Compute $\vec z_j^{k+1}=\vec z_j^k-\tau_k\!\left(\mat A_j^\T\vec v^{k+1}+\vec s^{k+1}\right)$

\textbf{Step 3:} Update $\tau_{k+1}$ so that $\tau_{k+1}>\tau_k$
\caption{Augmented Lagrangian method for \eqref{eq:cspineobj} \citep{zhang_efficient_2020}}
\label{alg:ssnal}
\end{algorithm}

\subsection{Implementation Details}
\label{sec:tuning}
The penalized nodewise regression problem \eqref{eq:cspineobj} involves two tuning parameters: $\lambda$ controlling elementwise sparsity and $\lambda_g$ controlling groupwise sparsity. We select $(\lambda, \lambda_g)$ via $k$-fold cross-validation.
Equivalently, we reparametrize the penalty using a global sparsity penalty $\lambda_0 > 0$ and a mixing parameter $\alpha_s \in [0,1]$,
so that the penalty in \eqref{eq:cspineobj} becomes
\[
g_j(\vec\gamma_j, \vec\beta_j) = \alpha_s\lambda_0 (\norm{\vec\gamma_j}_1 + \norm{\vec\beta_{j}}_1) + (1-\alpha_s)\lambda_0\norm{\vec\beta_{j, -0}}_{1,2}.
\]
For fixed $\alpha_s$, we compute solutions along a regularization path in $\lambda_0$. \robin{Following standard practice, we standardize the columns of the design matrix within each nodewise regression to ensure even penalization of the parameters.}

As in standard nodewise regression approaches to GGMs, each row of the precision matrix is estimated separately, and the resulting estimate is not necessarily symmetric \citep{meinshausen}.
In our setting, each nodewise regression in \eqref{eq:cspineobj} estimates one row of $\mat\Theta$, and asymmetry arises due to both independent row estimation and the different row scaling factors in \eqref{eq:reparam}.
To enforce symmetry, we apply a post hoc symmetrization step after performing all nodewise regressions.
Specifically, with a solution $\hvec\beta_j$ to \eqref{eq:cspineobj} and the corresponding error variance estimate in \eqref{eq:noisevar}, set
$\tvec\beta_j = -\hvec\beta_j / \hat\sigma^2_{\ep_j}$
and construct symmetric matrices $\tvec B_h \in \R^{p \times p}$ for all $h = 0, 1,\dotsc, q$ with entries
\[
    [\tvec B_h]_{jk} = [\tvec B_h]_{kj} = \tilde{\beta}_{jkh}
    \mathds{1} \bracks*{
    \abs{\tilde\beta_{jkh}} < \abs{\tilde\beta_{kjh}} } +
    \tilde{\beta}_{kjh}
    \mathds{1} \bracks*{
    \abs{\tilde\beta_{jkh}} > \abs{\tilde\beta_{kjh}}},
\]
where the expression $\mathds{1} \bracks*{P}$ is the indicator function of the argument $P$: it equals $1$ if $P$ is true and $0$ otherwise.
This corresponds to the ``and-rule'' symmetrization, under which $[\tilde{\mat B}_h]_{jk}$ is nonzero only if both $\tilde\beta_{jkh}$ and $\tilde\beta_{kjh}$ are nonzero.
An alternative is the ``or-rule,'' defined by
\[
    [\tvec B_h]_{jk} = [\tvec B_h]_{kj} = \tilde{\beta}_{jkh}
    \mathds{1}{\bracks*{
    \abs{\tilde\beta_{jkh}} \geq \abs{\tilde\beta_{kjh}} }} +
    \tilde{\beta}_{kjh}
    \mathds{1}{\bracks*{
    \abs{\tilde\beta_{jkh}} \leq \abs{\tilde\beta_{kjh}}}},
\]
which retains an edge if either coefficient is nonzero.
We adopt the more conservative ``and-rule'', noting that both approaches are asymptotically equivalent \citep{meinshausen} and have been widely used in the GGM literature \citep{cai2011constrained, liu2012tiger, Li_2021, tran2022completely, cai13, zhangli}.

Finally, the mean vector and precision matrix are recovered via \eqref{eq:reparam} as
\begin{align}
  \label{eq:predict}
\hvec\Omega(\vec u^{(i)}) &= \tvec B_0 + \sum_{h=1}^q \tvec B_h u^{(i)}_h \quad
\text{and} \quad
\hvec\mu(\vec u^{(i)}) = (\hvec\Omega(\vec u^{(i)}))^{-1} \diag(\hvec\Omega) \hvec\Gamma \vec u^{(i)}.
\end{align}

\robin{
As with other graphical model methods based on independent neighborhood selection, the resulting estimator $\hvec\Omega(\vec u^{(i)})$ in \eqref{eq:predict} may fail to be positive definite. One consequence is that $\hvec\mu(\vec u^{(i)})$ is not immediately estimable. In such cases, a post-hoc rescaling may be applied to ensure positive definiteness as described in \citet{zhangli}.
}

\section{Theoretical Properties}
\label{sec:theory}
We present theoretical guarantees for the estimator defined in \eqref{eq:cspineobj}, which characterize the accuracy of the recovered graph structure and its dependence on covariates.
We begin by introducing notation and stating assumptions on the data-generating process. We then present our main results, including $\ell_2$ estimation error bounds and support recovery guarantees, followed by a discussion comparing our results with those of \citet{zhangli}.

For two sequences of real numbers $a_n$ and $b_n$, we write $a_n \precsim b_n$
if $a_n = O(b_n)$, that is, there exist constants $C>0$ and $N>0$ so that $a_n < C b_n$
for all $n \geq N$. If $a_n \precsim b_n$ and $b_n\precsim a_n$, we write $a_n\asymp b_n$.
We write $a_n = o(b_n)$ if $\lim_{n\to\infty} a_n/b_n = 0$.
Our analysis holds for any $j \in [p]$. For notational simplicity, we suppress the subscript $j$ when referring to $\vec\gamma_j$, $\vec\beta_j$,  $\vec\ep_j$ and $\mat W_{-j}$ in this section.

Let $(\hat{\vec \gamma}, \hat{\vec \beta})$ be the solution to \eqref{eq:cspineobj} and
let $(\vec \gamma, \vec \beta)$ be the true parameters.
Denote by $S_\gamma$ and $S_\beta$ the support sets of $\vec \gamma$ and $\vec\beta$ respectively.
Further let $S_{\beta,g}$ index the active groups of $\vec\beta$, that is $S_{\beta,g} = \sets*{h: \vec\beta_{j,h} \neq \vec 0,\, h\in[q]}$.
Denote by $s_\gamma,\, s_\beta$ and $s_{\beta, g}$ the cardinalities of these sets.
For a square matrix $\mat M$, denote by $\lammin(\mat M)$ and $\lammax(\mat M)$ its minimum and maximum eigenvalue, respectively.
We require the following assumptions on the model \eqref{eq:multitask} with specification \eqref{eq:cspinespec}.
\begin{assumption}
\label{assump:bounded_u_element}
Suppose $\sets{\vec u^{(i)}}_{i=1}^n$ are i.i.d. mean zero random vectors with uniformly bounded entries: there exists a constant $M_u > 0$ such that $\abs{u^{(i)}_h} \leq M_u$ for all $h\in[q]$ and $i\in[n]$.
\end{assumption}
\begin{assumption}
\label{assump:subgaussian_x}
The responses $\vec x^{(i)}$ are elementwise sub-Gaussian with bounded sub-Gaussian norm: there exists a constant $K_X > 0$ such that $\norm{x^{(i)}_j}_{\psi_2} \leq K_X$ for all $j\in[p]$ and $i\in[n]$.
\end{assumption}
Let $[\mat U, \mat W] \in \R^{n \times (p (q+1) - 1)}$ denote the augmented design matrix in \eqref{eq:cspineobj}.
Define the population Gram matrix $\SUW = \EE\parens*{[\mU, \mW]^\T[\mU, \mW]/n}$.
\begin{assumption}
\label{assump:spectralboundSUW}
There exist constants $0 < \phi_0^\ast \leq \phi_1^\ast < \infty$ such that
\[
\phi_0^\ast \leq \lammin(\SUW) \leq \lammax(\SUW) \leq \phi_1^\ast.
\]
\end{assumption}

\begin{assumption}
\label{assump:sparsity_scale}
The sparsity levels satisfy
$(s_\gamma + s_\beta)\log(pq) = o(\sqrt{n}/\log n)$.
\end{assumption}

Assumptions~\ref{assump:bounded_u_element}--\ref{assump:subgaussian_x} control the tail behavior of the covariates, responses, and their pairwise interactions, thus facilitating their concentration behavior.
Similar assumptions are also imposed in \citet{zhangli}.
Assumption~\ref{assump:spectralboundSUW} ensures the augmented design matrix is well-conditioned, preventing excessive collinearity among predictors in the nodewise regression \eqref{eq:cspineobj}.
Assumption~\ref{assump:spectralboundSUW} is closely related to the bounded eigenvalue assumptions on $\Cov(\bm U)$ and $\Cov(\bm X \given \bm U)$ in \citet{zhangli}. In Appendix~\ref{sec:sufficientconditions}, we provide interpretable sufficient conditions under which Assumption~\ref{assump:spectralboundSUW} holds.
Finally, Assumption~\ref{assump:sparsity_scale} governs the scaling of sparsity levels with the sample size. Combined with Theorem~4.1 of \citet{kuchi}, it allows us to control sparse deviations between sample and population Gram matrices (see Lemmas~\ref{lem:kuchi} and~\ref{lem:operatorboundUW} in Appendix~\ref{appdx:proofs}).
A sufficient condition for Assumption~\ref{assump:sparsity_scale} is $\log(pq) = O(n^{1/6})$ and $s_\gamma + s_\beta = o(n^{1/6})$, as assumed in \citet{zhangli}.

In practice we consider estimates $(\hvec\gamma, \hvec\beta)$ of varying sparsities
and select from this pool of candidate models via cross-validation.
Define $\hat s_{\gamma}^\mathrm{max}$ and $\hat s_{\beta}^\mathrm{max}$ to be the maximum
sparsity of $\vec\gamma$ and $\vec\beta$ out of all candidate models, chosen so that
$s_\gamma < \hat s_{\gamma}^\mathrm{max}$ and $s_\beta < \hat s_{\beta}^\mathrm{max}$.
Our first theorem describes the $\ell_2$ estimation error of a solution to \eqref{eq:cspineobj}.
\begin{theorem}
\label{thm:l2error}
Under Assumptions \ref{assump:bounded_u_element}--\ref{assump:sparsity_scale}, with $(\hat s_{\gamma}^\mathrm{max} + \hat s_\beta^\mathrm{max})\log(pq) = O(\sqrt{n})$, and with
\begin{equation}
\label{eq:penalties}
\lambda = C\frac{\sigma_\ep}{\sqrt{n}}\parens*{
\frac{s_\gamma\log(eq/s_\gamma) + s_\beta\log(ep)}{s_\gamma + s_\beta} + \frac{s_{\beta, g}}{s_\gamma + s_\beta}\log(eq/s_{\beta, g})}^{1/2},
\quad
\lambda_g = \lambda\sqrt{\frac{s_\gamma + s_\beta}{s_{\beta, g}}},
\end{equation}
we have that
\[
\norm{\hvec{\gamma}-\vec\gamma}_2^2 + \norm{\hvec{\beta}-\vec\beta}_2^2 \precsim \frac{\sigma_{\ep}^2}{n} \left\{ s_\gamma\log(eq/s_\gamma) + s_\beta\log(ep) + s_{\beta,g}\log(eq/s_{\beta,g}) \right\}
\]
holds with probability at least
$1 - C_1\exp\bracks*{-C_2\{s_\gamma \log(eq/s_\gamma) + s_\beta\log(e p) + s_{\beta,g}\log(eq/s_{\beta,g})\}}$ where
$C$, $C_1$, and $C_2$ are positive constants.
\end{theorem}

We next establish support recovery guarantees, under an additional mutual incoherence condition.
Denote the $(k,\ell)$-th entry of $\SUW$ by $\SUW(k,\ell)$.
\begin{assumption}
\label{assump:coherence}
Define
$\tau_j = 1+\sqrt{(s_\gamma + s_\beta)/s_{\beta, g}}$.
Suppose
\[
\max_{k\neq \ell}\abs{\SUW(k,\ell)} \leq \frac{1}{c_0(1+8\tau_j)(s_\beta + s_\gamma)}
\]
for a constant $c_0 > 2/\phi^\ast_0$.
\end{assumption}
Assumption~\ref{assump:coherence} imposes a mutual incoherence condition \citep{Van2009On} on the augmented design matrix. It controls correlations between active and inactive coordinates and is essential for deriving $\ell_\infty$ error bounds necessary for support recovery.

\begin{theorem}
\label{thm:supportrecovery}
Suppose the assumptions of Theorem~\ref{thm:l2error}, together with Assumption~\ref{assump:coherence}, hold, and that $\log p \asymp \log q$. Further assume that
\[n \geq A_1 \parens*{s_{\gamma} \log(eq/s_{\gamma}) + s_{\beta}\log(ep) + s_{\beta,g}\log(eq/s_{\beta, g})}\]
for some constant $A_1 > 0$.
With the same tuning parameters $\lambda$ and $\lambda_g$ as in \eqref{eq:penalties},
we have that
\[
\max\sets*{\norm{\hvec\gamma - \vec\gamma}_\infty, \norm{\hvec\beta - \vec\beta}_\infty}
\leq
\frac{3}{\phi^\ast_0}(\lambda+\lambda_g)\parens*{1 + \frac{6(1+4\tau_j)^2}{(1+8\tau_j)(c_0\phi^\ast_0 - 2)}}
\]
holds with probability at least $1-C_3\exp(-C_4 \log p)$ for some positive constants $C_3, C_4$.
\end{theorem}
As in the standard lasso regression, with an additional minimal signal strength condition, the $\ell_\infty$ bound in Theorem~\ref{thm:supportrecovery} implies that the support of $\vec\gamma$ and $\vec\beta$ can be recovered with high probability by thresholding the estimated coefficients at an appropriate level.
\begin{corollary}
\label{cor:supportrecovery}
Define
\[
\kappa_0 = \frac{3}{\phi^\ast_0}(\lambda+\lambda_g)\parens*{1 + \frac{6(1+4\tau_j)^2}{(1+8\tau_j)(c_0\phi^\ast_0 - 2)}}.
\]
Let $\tilde S_\gamma = \sets*{k : \abs{\hat\gamma_k} > \kappa_0}$ and $\tilde S_\beta = \sets*{k : \abs{\hat\beta_k} > \kappa_0}$.
If
$\min\sets*{\min_{\ell\in S_\gamma} \abs{\gamma_\ell}, \min_{k\in S_\beta} \abs{\beta_k}}
\geq
2\kappa_0$,
then $\PP(S_\gamma = \tilde S_\gamma\; \text{and}\; S_\beta = \tilde S_\beta) \geq 1 - C_3\exp(-C_4\log p)$.
\end{corollary}

\subsection{Understanding the Theoretical Guarantees}
\label{sec:theoretical_discussion}
We conclude this section with a discussion of the $\ell_2$ estimation error and a comparison with the results of \citet{zhangli}.
The same insights apply to the support recovery guarantees in Theorem~\ref{thm:supportrecovery}.

In \citet{zhangli}, two rates of estimation error are established: an \emph{oracle} rate and a \emph{two-stage} rate.
The oracle rate is achieved when the true $\mat\Gamma$ in \eqref{eq:zhanglispec} is known, so that the centered responses can be constructed exactly and the estimation in \eqref{eq:zhanglinodewise_oracle} can be carried out directly.
The resulting rate is
\begin{equation}
  \label{eq:zhangli_oracle_rate}
  \norm{\hvec \eta_j - \vec\eta_j}_2^2 \precsim \frac{\sigma^2_{\ep_j}}{n}(s_\beta\log(ep) + s_{\beta,g}\log(eq/s_{\beta, g})),
\end{equation}
which holds with probability at least $1 - C_1\exp[-C_2\{s_\beta \log(ep) + s_{\beta,g}\log(eq/s_{\beta, g})\}]$ for some positive constants $C_1, C_2$.

In practice, however, $\mat\Gamma$ is unknown and must be estimated. As reviewed in Section~\ref{sec:zhangli}, \citet{zhangli} employ a two-stage procedure in which $\mat \Gamma$ is estimated via the misspecified lasso regression
\begin{equation}
\label{eq:zhanglistage1}
\hvec\gamma_j = \argmin_{\vec\gamma_j} \frac{1}{2n}\norm*{\vec x_j - \mat U\vec \gamma_j}_2^2 + \lambda_1\norm{\vec\gamma_j}_1,
\end{equation}
where $\lambda_1 \geq 0$,
and the resulting estimate is used to construct centered responses $\hat{\vec z}_j = \vec x_j - \mat U \hat{\vec \gamma}_j$, which are then used in the second-stage regression \eqref{eq:zhanglinodewise_oracle}.

Although the resulting two-stage estimator can attain the oracle rate \eqref{eq:zhangli_oracle_rate}, the associated probability bound depends critically on the tuning parameter $\lambda_1$ in \eqref{eq:zhanglistage1}.
Specifically, suppose that $\lambda_1 \propto \sqrt{\kappa_1\log q / n}$ for some $\kappa_1 > 0$.
Standard lasso results imply that the first-stage estimation error satisfies $\norm{\hvec z_j - \vec z_j}_2^2/n \precsim s_\gamma \cdot \kappa_1 \log q / n$ with probability at least $1 - 3\exp(-\kappa_1\log q)$.
To ensure that this error does not dominate the second-stage estimation, a union bound must be applied over all $p \times q$ columns of the design matrix. This leads to the overall success probability of at least $1 - C_3\exp[C_4\{\log p - (\kappa_1 - 1)\log q\}]$ for positive constants $C_3, C_4$.

This dependence introduces a nontrivial tradeoff: if $\kappa_1$ is too small, the concentration probability deteriorates. If $\kappa_1$ is too large, the first-stage penalty becomes overly aggressive, increasing the bias in the first-stage estimator and potentially inflating the overall error. In addition, when $\log p$ dominates $\log q$, the success probability may no longer decay to zero, further weakening the theoretical guarantees.

By contrast, our estimator matches the oracle rate without requiring a separate centering stage, thereby avoiding error propagation from the first stage and eliminating the need to tune an additional penalty parameter.
Furthermore, the oracle rate in \citet{zhangli} is established under the scaling conditions $\log(pq) = O(n^{1/6})$ and $s_\gamma + s_\beta = o(n^{1/6})$. In comparison, Assumption~\ref{assump:sparsity_scale} is strictly weaker and allows for a wider range of scaling regimes. It permits
$\log(pq)$ and $s_\gamma + s_\beta$ to grow at faster rates relative to $n$, e.g.,
$\log(pq) = O(n^{1/4} / \log(n))$ and $s_\gamma + s_\beta = o(n^{1/4})$, while still ensuring consistency.

These theoretical advantages in our framework stem from the joint convexity of \eqref{eq:cspineobj}, which enables simultaneous estimation of $(\vec\gamma, \vec\beta)$ without intermediate approximation steps. In contrast, the two-stage procedure in \citet{zhangli} requires additional control over the first-stage estimation error, leading to more restrictive conditions.
We provide additional empirical evidence of these advantages in Section~\ref{sec:simulation}, particularly in settings where accurate first stage estimation in \citet{zhangli} is challenging.

\section{Simulation Studies}
\label{sec:simulation}
\robin{
We evaluate our method \texttt{cspine} and competitors through extensive simulation studies.
Our implementation of the optimization algorithm described in Section~\ref{sec:opt} to solve sparse-group lasso problems is based on the MATLAB package of \citet{zhang_efficient_2020}.
}

\subsection{Data Generation}
\label{sec:data_gen}
For each simulation setting, we generate a sparse coefficient matrix $\mat\Gamma \in \R^{p \times q}$ by independently setting each entry $\Gamma_{jk}$ equal to $0.15$ with probability $0.1$.
We generate the baseline precision matrix $\mat B_0$ using a preferential attachment algorithm \citep{barabasi} with power parameter equal to $1$. \robin{Among the $q$ covariates, half are drawn from $\Bernoulli(0.5)$ and half from $\mathrm{Unif}(0,1)$. The continuous covariates are then centered and standardized to have unit variance.} To introduce covariate-dependent effects, we select $q_e = 5$ covariates to have nonzero influence and construct the matrices $\mat B_h$ for $h \in [q_e]$ as Erd\H{o}s--R\'enyi graphs \citep{erdos} with edge probability $v_e = 0.01$. This choice ensures that the covariate-dependent graphs are sparser than the population graph \citep{clauset}.

Given the graph structure, the nonzero entries of ${\mat B_h}$ are generated independently from the distribution $\mathrm{Unif}([-0.5,-0.35] \cup [0.35,0.5])$. To ensure diagonal dominance of the resulting precision matrix, each row $j$ is scaled by
$\sum_h\sum_{k\neq j}\abs{\beta_{jkh}} \times 1.5$ followed by symmetrization of each $\mat B_h$ via averaging $\beta_{jkh}$ and $\beta_{kjh}$. This construction mirrors the data-generating mechanism in \citet{zhangli}.

\subsection{Model Specification and Experimental Setup}
\robin{For each observation $i \in [n]$, we construct the precision matrix $\mat\Omega^{(i)}$ via
$\mat B_0 + \sum_{h=1}^q \mat B_h u_h^{(i)}$, setting all entries of $\diag(\mat \Omega^{(i)}) $ to be $1$, so that $\sigma^2_{\ep_j} = 1$.
To construct the mean, we introduce $\delta \in [0,1]$ and set
\begin{equation}
\label{eq:delta}
\vec\mu(\vec u^{(i)}) = (1-\delta) \mat\Gamma \vec u^{(i)} + \delta \mat\Sigma(\vec u^{(i)}) \mat\Gamma \vec u^{(i)}.
\end{equation}
The parameter $\delta$ controls the relative influence of direct covariate effects versus covariance-mediated network effects. Note that $\delta = 0$ results in data sets where \texttt{RegGMM} is correctly specified while $\delta = 1$ means \texttt{cspine} is correctly specified. Both models are misspecified for intermediate values of $\delta$.
Finally we generate responses via $\vec x^{(i)} \sim N_p(\vec\mu^{(i)}, \mat\Sigma^{(i)})$ for $i \in [n]$.
}

\robin{
We fix $p = 25$ responses and consider $q= 50$ and $q = 100$ covariates.
For each setting, we generate $100$ independent data sets with sample sizes $n = 200$ and $n = 400$.
We compare against the nodewise regression method of \citet{zhangli} (referred to as \texttt{RegGMM}) as well as the multi-task extension in \citet{zhang-multi-task-2025} (\texttt{mt-RegGMM}).
We found that the method of \citet{wang_high-dimensional_2025} has intractable runtime in the large $q$ setting and therefore exclude it from comparison.
When $\delta = 0$ in \eqref{eq:delta}, we additionally compare against an oracle estimator (\texttt{RegGMM-oracle}) where the true $\mat\Gamma$ is plugged into \texttt{RegGMM}; the performance of \texttt{RegGMM-oracle} reflects the baseline oracle rate in \eqref{eq:zhangli_oracle_rate}.
The tuning parameters for all methods were selected using $5$-fold cross-validation over a path of $100$ penalties.
For \texttt{RegGMM} and \texttt{mt-RegGMM}, the stage 1 step uses $10$-fold cross-validation to select $\lambda_1$ over a path of $100$ penalties in \eqref{eq:zhanglistage1}.
}
We report the mean and standard error of the following performance metrics:
$\text{TPR}$ (true positive rate of detected edges across all $\sets{\mat B_h}$), $\text{FPR}$ (false positive rate of detected edges across all $\sets{\mat B_h}$), $\vec\beta_{\mathrm{err}} = \sum_{j=1}^p \norm{\hat{\vec\beta}_j - \vec\beta_j}_2$ (nodewise estimation error), and $\mat\Omega_{\mathrm{err}} = \sum_{i=1}^n \norm{\hat{\mat\Omega}^{(i)} - \mat\Omega^{(i)}}^2_{F,\mathrm{off}} / n$ (precision matrix estimation error over off-diagonal entries).

\subsection{Results}
Table~\ref{tbl:sim_orig} shows the results when $\delta = 0$ so that the true mean is $\vec\mu(\vec u) = \mat\Gamma \vec u$.
Comparing against \texttt{RegGMM-oracle}, we can see the effect of estimation error of $\hvec\Gamma$ for the two-stage methods. Despite the misspecification, \texttt{cspine} maintains competitive performance owing to the jointly convex formulation.
Table~\ref{tbl:sim} presents results under $\delta = 1$ so that $\vec\mu(\vec u) = \mat\Sigma(\vec u)\mat\Gamma\vec u$. As expected, \texttt{cspine} outperforms competitors as it correctly specifies the network effect of covariates on the mean. The stage 1 misspecification of \texttt{RegGMM} and \texttt{mt-RegGMM} leads to errors in constructing the centered responses and these errors propagate to errors in estimating the precision matrix.

\begin{table}[tbp]
\centering
\small
\begin{tabular}{r|r|r|cccc}
  \hline
$n$ & $q$ & Method & $\text{TPR}$ & $\text{FPR}$ & $\vec\beta_\text{err}$ & $\mat\Omega_\text{err}$ \\
  \hline
 & \multirow{4}{*}{$50$} & \texttt{cspine} & $0.747$ $(0.088)$ & $\mathbf{0.009}$ $(0.002)$ & $\mathbf{1.665}$ $(0.122)$ & $\mathbf{0.058}$ $(0.003)$ \\
 &  & \texttt{RegGMM} & $\mathbf{0.748}$ $(0.082)$ & $\mathbf{0.009}$ $(0.002)$ & $1.676$ $(0.120)$ & $\mathbf{0.058}$ $(0.003)$ \\
 &  & \texttt{mtRegGMM} & $0.684$ $(0.081)$ & $0.015$ $(0.002)$ & $1.877$ $(0.118)$ & $0.061$ $(0.003)$ \\
  \cline{3-7}
 &  & \texttt{RegGMM-oracle} & $0.774$ $(0.080)$ & $0.009$ $(0.002)$ & $1.595$ $(0.130)$ & $0.055$ $(0.003)$ \\
  \cline{2-7}
200 & \multirow{4}{*}{$100$} & \texttt{cspine} & $\mathbf{0.487}$ $(0.083)$ & $\mathbf{0.004}$ $(0.001)$ & $\mathbf{1.656}$ $(0.080)$ & $\mathbf{0.059}$ $(0.003)$ \\
 &  & \texttt{RegGMM} & $0.477$ $(0.085)$ & $\mathbf{0.004}$ $(0.001)$ & $1.660$ $(0.081)$ & $\mathbf{0.059}$ $(0.003)$ \\
 &  & \texttt{mtRegGMM} & $0.466$ $(0.082)$ & $0.010$ $(0.001)$ & $1.807$ $(0.091)$ & $0.060$ $(0.002)$ \\
  \cline{3-7}
 &  & \texttt{RegGMM-oracle} & $0.548$ $(0.089)$ & $0.004$ $(0.001)$ & $1.543$ $(0.094)$ & $0.055$ $(0.003)$ \\
  \hline
 & \multirow{4}{*}{$50$} & \texttt{cspine} & $0.916$ $(0.045)$ & $0.013$ $(0.002)$ & $1.315$ $(0.105)$ & $0.044$ $(0.002)$ \\
 &  & \texttt{RegGMM} & $\mathbf{0.923}$ $(0.044)$ & $\mathbf{0.011}$ $(0.002)$ & $\mathbf{1.273}$ $(0.109)$ & $\mathbf{0.043}$ $(0.002)$ \\
 &  & \texttt{mtRegGMM} & $0.882$ $(0.048)$ & $0.016$ $(0.002)$ & $1.473$ $(0.106)$ & $0.046$ $(0.002)$ \\
  \cline{3-7}
 &  & \texttt{RegGMM-oracle} & $0.927$ $(0.039)$ & $0.011$ $(0.002)$ & $1.219$ $(0.107)$ & $0.041$ $(0.002)$ \\
  \cline{2-7}
400 & \multirow{4}{*}{$100$} & \texttt{cspine} & $0.717$ $(0.073)$ & $0.008$ $(0.001)$ & $1.395$ $(0.094)$ & $0.048$ $(0.002)$ \\
 &  & \texttt{RegGMM} & $\mathbf{0.728}$ $(0.074)$ & $\mathbf{0.006}$ $(0.001)$ & $\mathbf{1.333}$ $(0.104)$ & $\mathbf{0.046}$ $(0.003)$ \\
 &  & \texttt{mtRegGMM} & $0.671$ $(0.070)$ & $0.012$ $(0.001)$ & $1.503$ $(0.091)$ & $0.049$ $(0.002)$ \\
  \cline{3-7}
 &  & \texttt{RegGMM-oracle} & $0.771$ $(0.069)$ & $0.005$ $(0.001)$ & $1.234$ $(0.107)$ & $0.042$ $(0.003)$ \\
  \hline
\end{tabular}
\caption{Mean and standard error of performance metrics over $100$ replications with $p=25$ responses, $q$ covariates, $n$ observations, and covariate-adjusted network structure $\mat\Omega(\vec u) = \mat B_0 + \sum_{h=1}^q {\mat B_h u_h}$ and mean $\vec\mu(\vec u) = \mat\Gamma \vec u$.}
\label{tbl:sim_orig}
\end{table}

\begin{table}[tbp]
\centering
\small
\begin{tabular}{r|r|r|cccc}
  \hline
$n$ & $q$ & Method & $\text{TPR}$ & $\text{FPR}$ & $\vec\beta_\text{err}$ & $\mat\Omega_\text{err}$ \\
  \hline
 & \multirow{3}{*}{$50$} & \texttt{cspine} & $\mathbf{0.839}$ $(0.068)$ & $\mathbf{0.007}$ $(0.002)$ & $\mathbf{1.432}$ $(0.135)$ & $\mathbf{0.049}$ $(0.003)$ \\
 &  & \texttt{RegGMM} & $0.797$ $(0.076)$ & $0.009$ $(0.002)$ & $1.552$ $(0.133)$ & $0.053$ $(0.003)$ \\
 &  & \texttt{mtRegGMM} & $0.742$ $(0.083)$ & $0.014$ $(0.002)$ & $1.759$ $(0.134)$ & $0.057$ $(0.003)$ \\
  \cline{2-7}
200 & \multirow{3}{*}{$100$} & \texttt{cspine} & $\mathbf{0.679}$ $(0.081)$ & $\mathbf{0.003}$ $(0.001)$ & $\mathbf{1.317}$ $(0.118)$ & $\mathbf{0.047}$ $(0.003)$ \\
 &  & \texttt{RegGMM} & $0.583$ $(0.089)$ & $0.004$ $(0.001)$ & $1.497$ $(0.119)$ & $0.053$ $(0.003)$ \\
 &  & \texttt{mtRegGMM} & $0.561$ $(0.088)$ & $0.009$ $(0.001)$ & $1.643$ $(0.122)$ & $0.054$ $(0.003)$ \\
  \hline
 & \multirow{3}{*}{$50$} & \texttt{cspine} & $\mathbf{0.953}$ $(0.032)$ & $\mathbf{0.009}$ $(0.002)$ & $\mathbf{1.066}$ $(0.115)$ & $\mathbf{0.035}$ $(0.003)$ \\
 &  & \texttt{RegGMM} & $0.938$ $(0.039)$ & $0.010$ $(0.002)$ & $1.178$ $(0.116)$ & $0.039$ $(0.003)$ \\
 &  & \texttt{mtRegGMM} & $0.909$ $(0.048)$ & $0.014$ $(0.003)$ & $1.371$ $(0.128)$ & $0.043$ $(0.003)$ \\
  \cline{2-7}
400 & \multirow{3}{*}{$100$} & \texttt{cspine} & $\mathbf{0.852}$ $(0.053)$ & $\mathbf{0.005}$ $(0.001)$ & $\mathbf{1.023}$ $(0.096)$ & $\mathbf{0.035}$ $(0.002)$ \\
 &  & \texttt{RegGMM} & $0.783$ $(0.067)$ & $\mathbf{0.005}$ $(0.001)$ & $1.183$ $(0.112)$ & $0.041$ $(0.003)$ \\
 &  & \texttt{mtRegGMM} & $0.705$ $(0.064)$ & $0.011$ $(0.002)$ & $1.360$ $(0.104)$ & $0.044$ $(0.003)$ \\
  \hline
\end{tabular}
\caption{Mean and standard error of performance metrics over $100$ replications with $p=25$ responses, $q$ covariates, $n$ observations, and covariate-adjusted network structure $\mat\Omega(\vec u) = \mat B_0 + \sum_{h=1}^q {\mat B_h u_h}$ and covariance-mediated mean $\vec\mu(\vec u) = \mat\Sigma(\vec u)\mat\Gamma \vec u$.}
\label{tbl:sim}
\end{table}

Our next study compares the methods across $100$ replications of the intermediate settings $\delta = 0.25, 0.5, 0.75$ in \eqref{eq:delta}, so that the influence of $\vec u$ on $\vec\mu(\vec u)$ is a mix of direct and indirect (network dependent) effects. The resulting $\vec\beta_{\mathrm{err}}$ are plotted in Figure~\ref{fig:vary-delta}.
We see that \texttt{cspine} performs similarly to \texttt{RegGMM} when $\delta = 0$ (so \texttt{RegGMM} is correctly specified), yet the gap between the two methods increases favorably for \texttt{cspine} as $\delta \to 1$. The gap shows that in the intermediate settings where all methods are misspecified, the jointly convex formulation outperforms the two-stage approach where the performance is more sensitive to estimation error in the first stage.

\begin{figure}[tbp]
\centering
\includegraphics[width = 0.7\linewidth]{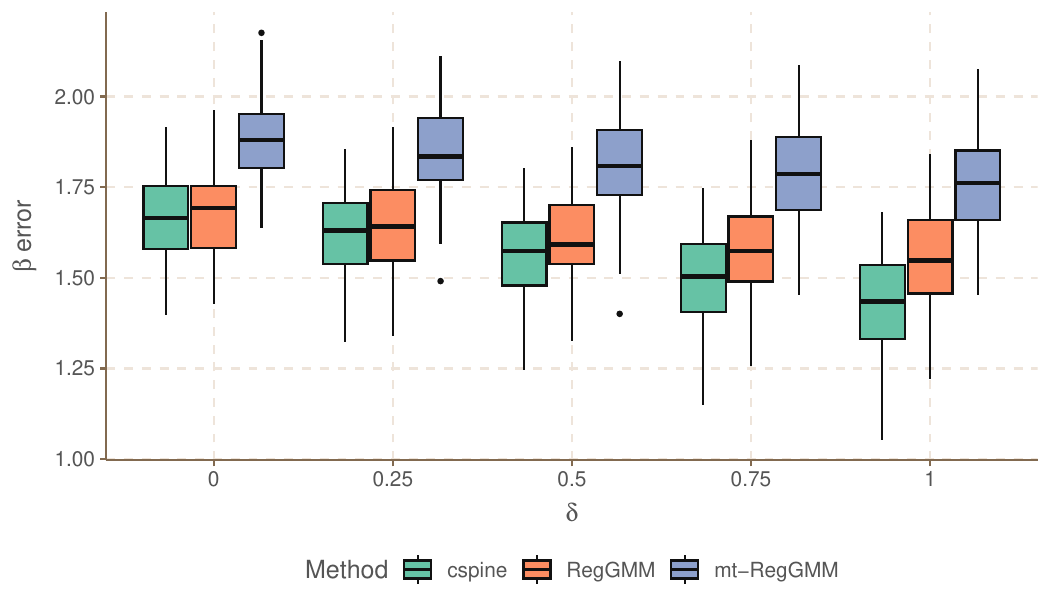}
\caption{Estimation error of $\sets{\mat B_h}$ for $p = 25$, $q=50$, $n=200$ and $\delta = 0, 0.25, 0.5, 0.75, 1$ in \eqref{eq:delta} over $100$ replications for each setting.  \texttt{RegGMM} and \texttt{mt-RegGMM} are correctly specified when $\delta = 0$ while \texttt{cspine} is correctly specified when $\delta = 1$. All methods are misspecified for intermediate values of $\delta$. As network effects on the mean dominate, the performance of \texttt{cspine} improves over that of two-stage methods.}
\label{fig:vary-delta}
\end{figure}

Finally we investigate the performance of \texttt{cspine} and \texttt{RegGMM} as the signal-to-noise ratio varies in the setting of \eqref{eq:zhanglispec}.
We define the signal-to-noise ratio as $\snr(\vec u) = \EE(\norm{\mat\Gamma \vec u}_2^2) / \EE[\tr\sets{\mat\Sigma(\vec u)}]$. To vary the $\snr$, before generating the responses we sample $q/2$ binary and continuous covariates as in Section~\ref{sec:data_gen} and then multiply the covariates by a constant $\varsigma$ such that the average $\snr(\varsigma \cdot \vec u)$ over $n$ samples is a target quantity. The responses are then drawn as before with $\delta=0$ in \eqref{eq:delta}. Figure~\ref{fig:vary-snr} shows the results with $p=25$, $q=50$, $n=200$, and target $\snr = 0.01, 0.025, 0.05, 0.075, 0.1$, with $50$ replications for each setting.
The results show that \texttt{cspine} has better performance than \texttt{RegGMM} in low $\snr$ settings, with performance equalizing around $\snr = 0.05$, and tilting in favor of \texttt{RegGMM} for larger $\snr$.
An explanation is that a low $\snr$ impacts \texttt{RegGMM} in both stages, resulting in noisier stage 1 residuals that compound the noisy graphical regression in stage 2.
As $\snr$ increases, $\hvec\Gamma$ becomes accurate enough that the contamination in stage 2 is negligible, and the advantage of correct specification under \eqref{eq:zhanglispec} allows \texttt{RegGMM} to overtake \texttt{cspine}.

\begin{figure}[tbp]
\centering
\includegraphics[width = 0.7\linewidth]{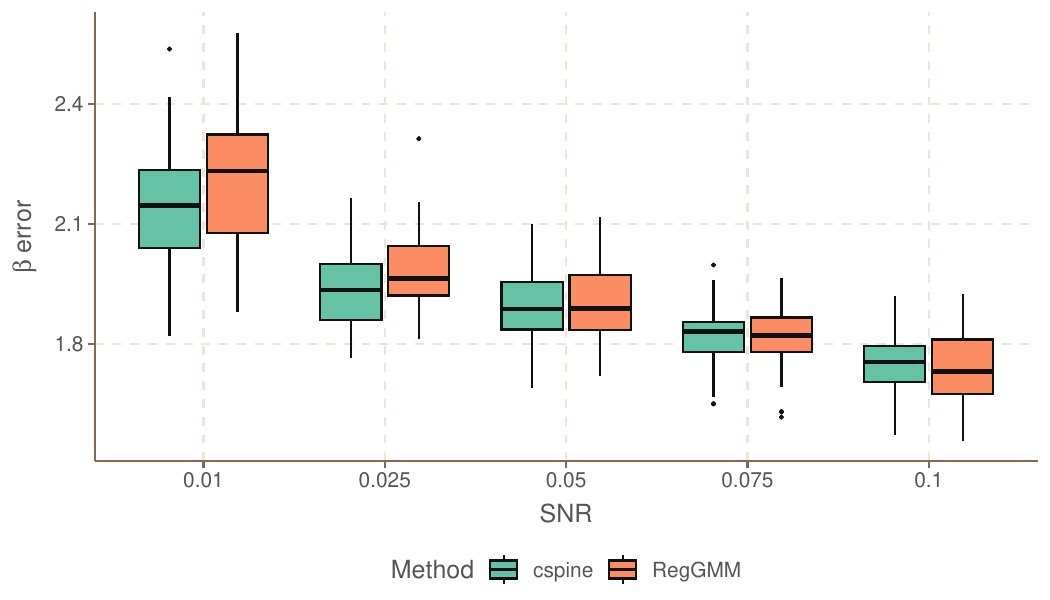}
\caption{Estimation error of $\sets{\mat B_h}$ for $p = 25$, $q=50$, $n=200$ and varying average $\snr$ over $50$ replications for each setting. The mean is given by $\delta=0$ in \eqref{eq:delta} so that \texttt{RegGMM} is correctly specified. \texttt{cspine} has lower estimation error than \texttt{RegGMM} in low $\snr$ regimes despite being misspecified.}
\label{fig:vary-snr}
\end{figure}

In Table~\ref{tbl:timing} we show the mean and standard deviation runtime for the nodewise regressions using Algorithm~\ref{alg:ssnal} for the $n=200$, $p=25$, and $q=50$ and $q = 100$ settings. The runtime is the wall clock time to compute $5$-fold cross-validation over $100$ penalty parameters for each method, averaged over $p$ responses, with each regression run on a 2.40GHz Intel Xeon processor. We see that the runtime of \texttt{cspine} is only marginally more than that of \texttt{RegGMM} since its nodewise regressions contain just $q$ more parameters.

\begin{table}[tbp]
\centering
\begin{tabular}{c|cc}
\hline
$q$ & \texttt{RegGMM} & \texttt{cspine} \\
\hline
$50$ & $4.03$ ($0.32$)     &  $4.06$ $(0.35)$ \\
$100$ & $7.62$ ($3.64$)    &  $7.96$ $(3.75)$ \\
\hline
\end{tabular}
\caption{\label{tbl:timing}Mean (and standard deviation) runtime in seconds to compute each nodewise regression with $p=25$ and $n=200$ for \texttt{RegGMM} and \texttt{cspine}, including $5$-fold cross-validation over $100$ penalty parameters.}
\end{table}

\section{Data Application}
\label{sec:realdata}
We demonstrate our method on two real data sets. The first is a glioblastoma multiforme (GBM) eQTL study, where we reanalyze data from \citet{zhangli} to illustrate how \texttt{cspine} recovers SNP-modified gene co-expression networks. The second is a dietary survey study from the American Gut Project, where we examine how dietary factors affect the conditional dependence structure of gut microbiome abundances.

\subsection{GBM eQTL Analysis}
Glioblastoma multiforme (GBM) is the most malignant type of brain cancer and patient prognosis is typically very poor.
Although there has been research on the genetic signaling pathways involved in the proliferation of GBM, it remains largely incurable; see \citet{glio-hanif} for a survey.
It is important to understand the conditional independence structure of genes involved in GBM in order to discover new drug therapies \citep{glio-kwiakowska}.
Our estimated graphs describe the conditional independence of co-expressions in a gene network; hence, we refer to estimated networks and effects of SNPs on this network.

We reanalyze a GBM eQTL data set that was reported in \citet{zhangli}; \robin{we are grateful to Emma Zhang for making this data set available to us.}
The data set contains microarray and SNP profiling data of $n=401$ GBM patients from the REMBRANDT trial (GSE108476).
We use the expression levels of $p=73$ genes known to belong to the human glioma pathway according to the Kyoto Encyclopedia of Genes and Genomes (KEGG) database \citep{KEGG}.
We also consider $q=118$ SNPs that are local to these $73$ genes. The SNPs are binary-coded, with $0$ indicating homozygous major alleles at that locus and $1$ otherwise.

Figure~\ref{fig:pop} displays the estimated population network.
The highly connected nodes include \textit{MTOR}, a critically dysregulated signaling hub and \textit{MDM2}, a highly mutated oncogene \citep{zhao_regulation_2014}. Known GBM signaling pathways are present in Figure~\ref{fig:pop} such as connections between the genes in the \textit{PDGRF/PI3K/AKT/mTOR} pathway \citep{network}.

\begin{figure}[tbp]
\centering
\includegraphics[width = 0.7\linewidth]{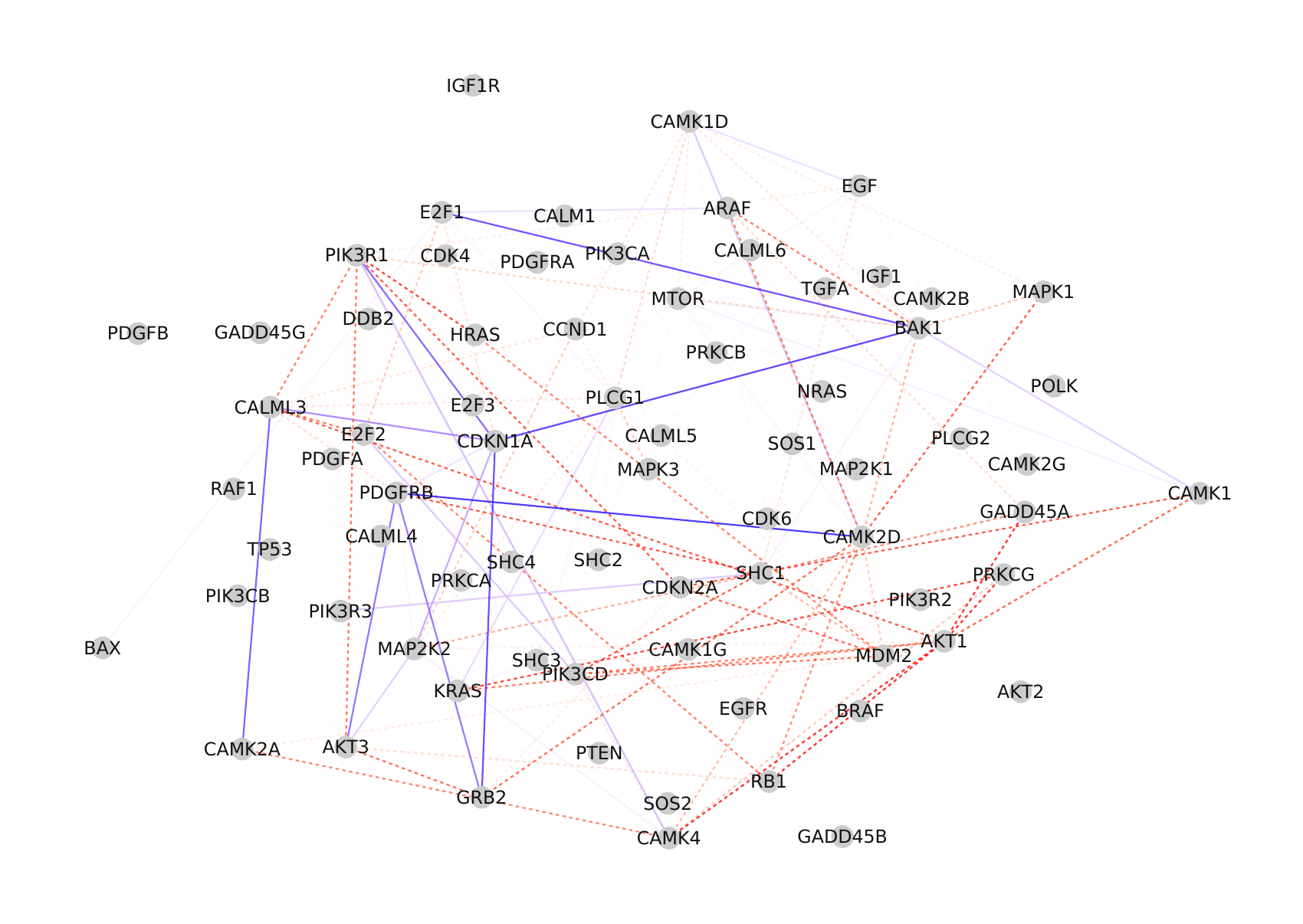}
\caption{Population network from GBM eQTL data estimated by \texttt{cspine}. The graph structure is determined by the estimates of $\mat B_0$ in \eqref{eq:cspinespec}. Solid blue lines and dashed red lines indicate positive and negative edge weights, respectively.}
\label{fig:pop}
\end{figure}

Our method identifies $9$ SNPs that potentially modify expression patterns in the network. We highlight two representative examples of these SNPs and their estimated effects on the network.
The first is \texttt{rs1267622}, a variant local to the \textit{BRAF} gene, whose estimated effects are shown in Figure~\ref{fig:remb_cov} (\textit{left}). It may mediate the co-expressions of the pairs \textit{SHC3}--\textit{PIK3R2} and \textit{CAMK2D}--\textit{GADD45A}.
This is a plausible finding because \textit{BRAF} is a key component of the \textit{Ras/Raf/MEK/ERK} signaling pathway, and \textit{SHC3} is an adaptor protein that sits upstream of both this pathway and the \textit{PI3K/AKT} pathway, linking receptor activation to \textit{PI3K} regulatory subunits such as \textit{PIK3R2} \citep{network}. A variant near \textit{BRAF} may therefore shift the coordination between these two branches. Similarly, \textit{CAMK2D} and \textit{GADD45A} are both involved in cellular stress response; \textit{BRAF} variants that reroute signaling toward stress kinase cascades may affect the co-expression of these two genes.
Meanwhile, \texttt{rs759950} (Figure~\ref{fig:remb_cov} \textit{right}), located near \textit{CAMK4}, may mediate the co-expressions of \textit{PIK3CA} and \textit{AKT1}, key components of the \textit{PI3K/AKT} pathway, mutations of which are known to be oncogenic, with targeted inhibition being an active area of cancer research \citep{samuels, batsios_pi3kmtor_2019}.
Calcium signaling via \textit{CAMK4} may play a role in this important pathway.

Our method additionally estimates that \texttt{rs6939054} may affect \textit{CAMK1--GADD45A} and \textit{CAMK4--GADD45A}, \texttt{rs10519200} may affect \textit{AKT1--RB1}, \texttt{rs1347069} may affect \textit{CAMK1--AKT1}, \texttt{rs10518690} may affect \textit{CAMK1--CAMK4}, \texttt{rs10491204} may affect \textit{AKT1--AKT2}, \texttt{rs9303504} may affect \textit{AKT1--CAMK1G}, \texttt{rs10515427} may affect \textit{AKT1--BRAF}, and \texttt{rs2877260} may affect \textit{GADD45A--RB1}.

\begin{figure}[tbp]
\centering
\includegraphics[width = \linewidth]{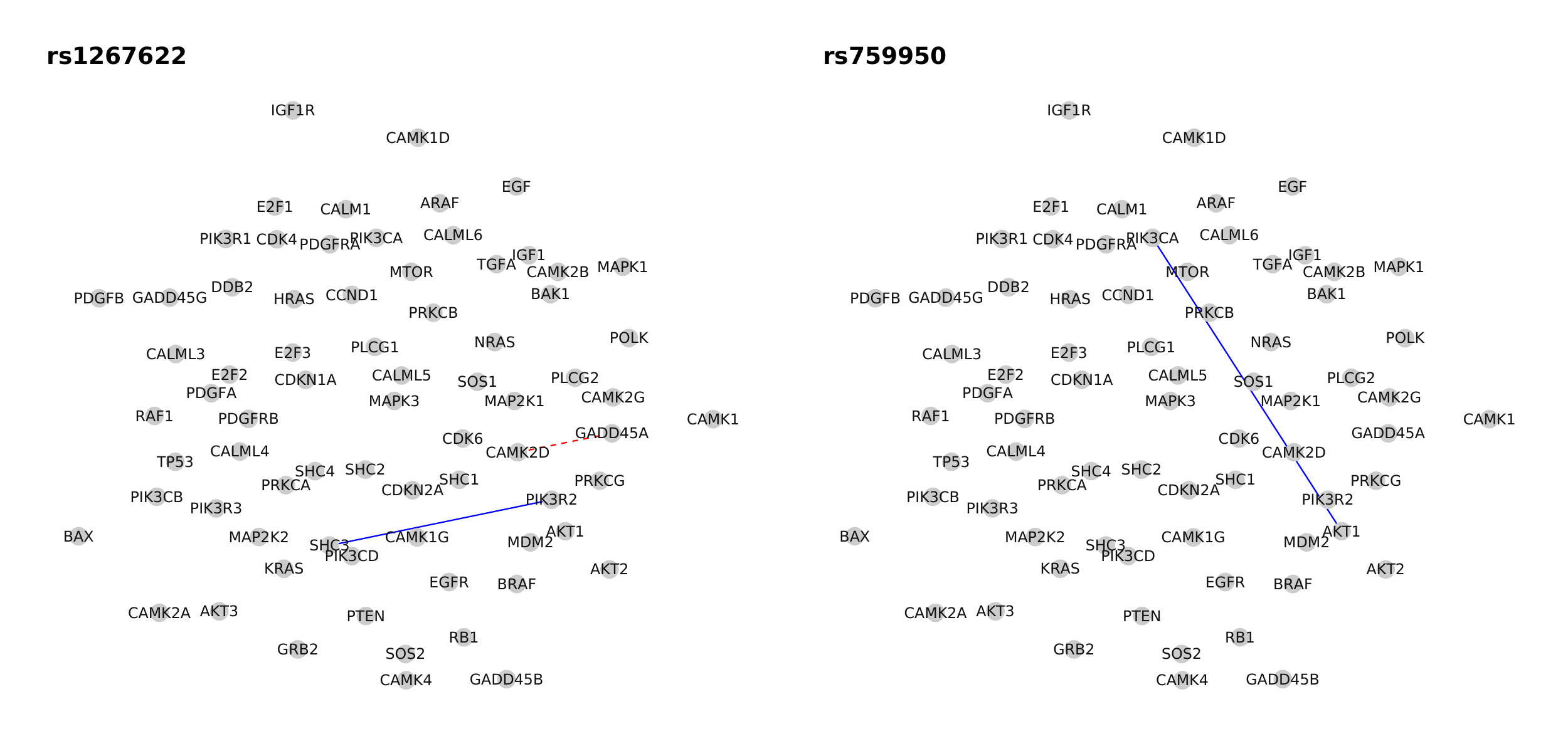}
\caption{Estimated covariate networks for two SNPs \texttt{rs1267622} (\textit{left}) and \texttt{rs759950} (\textit{right}).
The graph structure corresponds to the sparsity pattern of the matrix $\mat B_h$ corresponding to the SNPs in \eqref{eq:cspinespec}.
Solid blue lines and dashed red lines indicate positive and negative effects of this SNP on the partial correlation of co-expressions.}
\label{fig:remb_cov}
\end{figure}

\subsection{AGP Microbiome Analysis}
The human gastrointestinal tract is an extraordinarily dense and dynamic microbial ecosystem, hosting trillions of microorganisms that have co-evolved with the host to perform essential physiological, metabolic, and immunological functions.
These microorganisms are intimately linked by cross-feeding, where metabolites produced by one microbe are consumed by another, resulting in interdependent microbial communities; see \citet{culp_cross-feeding_2023} for a review.
Moreover, recent research suggests that diet plays a major role in the composition of this ecosystem; see for example \citet{conlon_impact_2014,klimenko_hallmarks_2022,snodgrass_butyrate-producing_2026} among many others.
It is of clinical and biological interest to understand how dietary factors may affect the complex interrelationships of human gut bacteria.

We analyze data from the American Gut Project (AGP), a large, crowd-funded study of the human microbiome \citep{mcdonald_american_2018}. AGP participants submit samples and fill out detailed lifestyle and diet surveys. 16S rRNA sequencing is applied to the samples, operational taxonomic units (OTUs) are identified, and the OTU abundances are recorded.
Estimating the inverse covariance of OTU abundances has been applied to AGP data before in \citet{kurtz_sparse_2015}. Here we provide a brief analysis of the effect of diet on the OTU abundance network using subject-level dietary survey responses as covariates.
We selected $p=20$ OTUs to analyze, representing the top $20$ most abundant OTUs in our data.
The $q=11$ covariates are binarized responses indicating the frequency of consuming various types of grains, seafood, fruits, vegetables, meats, sugar, and alcohol, with responses ``rarely'' and ``never'' mapped to $0$ and all other responses to $1$.
Subsetting to ``omnivores'' only, our final data set consisted of $n=367$ subjects with no missing data.
For each subject we followed the standard procedure of converting raw OTU counts into relative abundances followed by a center log ratio transform \citep{kurtz_sparse_2015}.

Figure~\ref{fig:agp_pop} shows the estimated population network from our proposed method. The nodes represent OTUs and are labeled with an ID along with its associated taxonomical family.
The OTU~4478125 is highly connected in this graph. It corresponds to the species \textit{Faecalibacterium prausnitzii}, which is one of the most abundant bacteria in the human microbiome \citep{sabater_mgem_2025}.
We find that it is highly connected to OTUs in the genus \textit{Bacteroides}.
This is an interesting finding as \textit{F. prausnitzii} has been shown to stabilize the gut environment and benefit the physiology of the host \citep{snodgrass_butyrate-producing_2026}. The cross-feeding relationships between \textit{F. prausnitzii} and members of different genera, including \textit{Bacteroides}, have also been observed before \citep{rios-covian_enhanced_2015, wrzosek_bacteroides_2013}.
Another interesting observation is the high connectivity of OTU~4365130, identified with the species \textit{Parabacteroides distasonis}, which recent research suggests to play a key role in human health \citep{ezeji_parabacteroides_nodate}.

\begin{figure}[tbp]
\centering
\includegraphics[width = 0.6\linewidth]{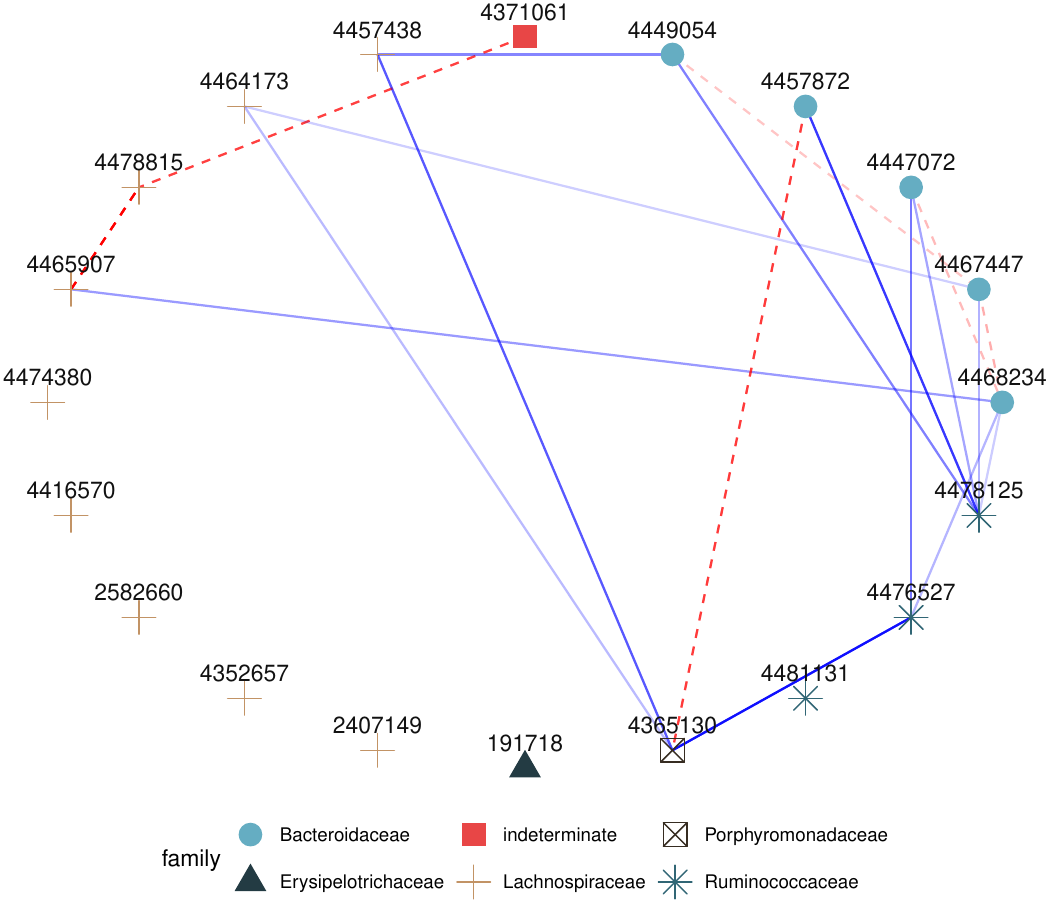}
\caption{Estimated population network for relative OTU abundances from AGP data. Solid blue lines and dashed red lines indicate positive and negative edge weights, respectively.}
\label{fig:agp_pop}
\end{figure}

Of the $11$ dietary covariates, we identified two to have non-zero effects: red meat consumption and alcohol. Figure~\ref{fig:agp_cov} shows the estimated effects. Red meat consumption may affect the relative abundance of \textit{Bacteroides} OTU~4457872 with OTUs 4478125 and 4476527, both identified with \textit{F. prausnitzii}, as well as a member of the \textit{Lachnospiraceae} family.
Alcohol consumption may affect the relationship between \textit{F. prausnitzii} and a member of \textit{Lachnospiraceae}.

\begin{figure}[tbp]
\centering
\includegraphics[width = \linewidth]{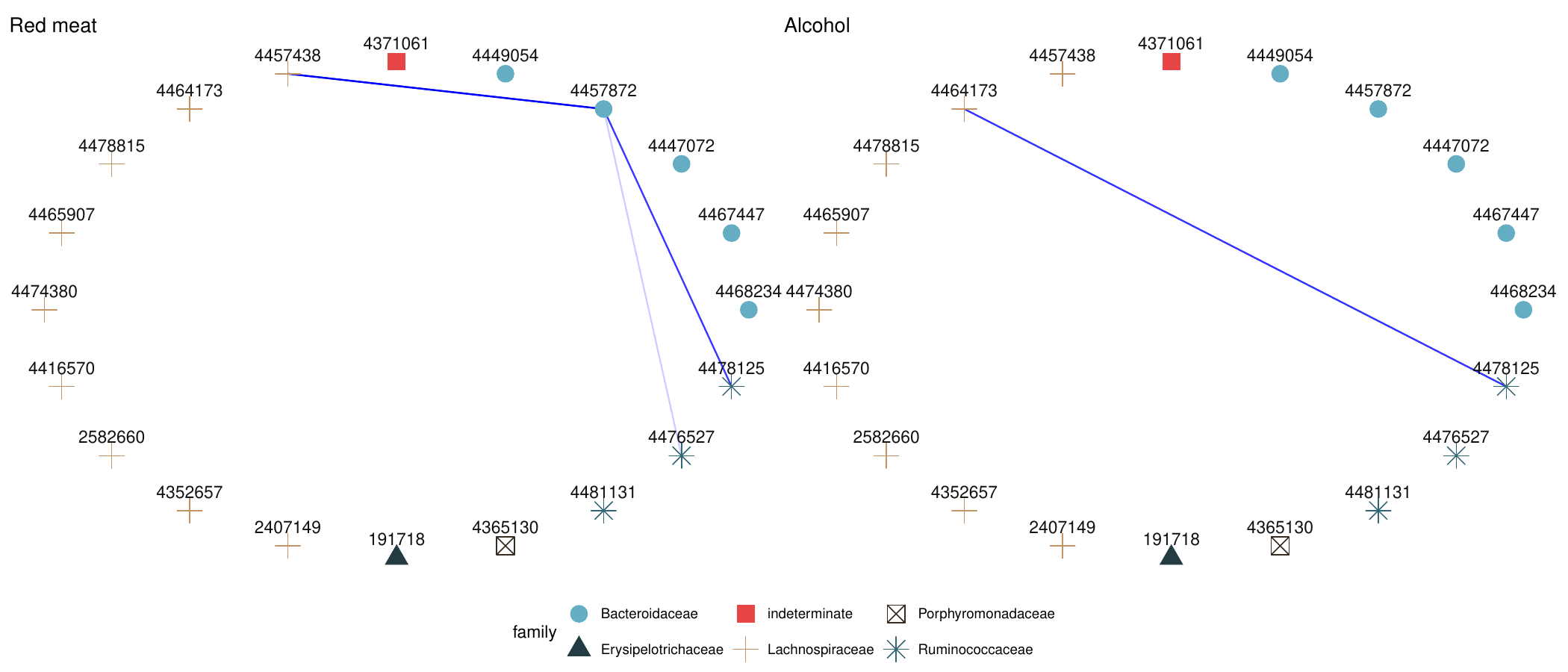}
\caption{Covariate network effects of red meat (\textit{left}) and alcohol (\textit{right}) on relative OTU abundances estimated from AGP data.}
\label{fig:agp_cov}
\end{figure}

\section{Discussion}
\label{sec:discussion}
In this work, we propose a new framework for covariate-adjusted Gaussian graphical models based on a natural parametrization that yields a jointly convex formulation.
This approach enables unified estimation of both the covariate-dependent mean and covariate-dependent precision matrix, in contrast to the two-stage procedure of \citet{zhangli}.
\robin{Furthermore, the natural parametrization is more than a computational convenience. The resulting mean structure $\vec\mu(\vec u) = \mat\Sigma(\vec u)\,\diag(\mat\Omega)\,\mat\Gamma\vec u$ describes an interpretable mechanism in which covariates influence the mean response through broader network effects.}
Our theoretical analysis shows that the convex formulation permits more relaxed scaling conditions on the sparsity level of the true parameter relative to the sample size $n$. These advantages are supported by our simulation results, which demonstrate improved performance in regimes where stagewise methods are challenged by error propagation.

The proposed framework is flexible and does not impose restrictions on the form of the penalty functions $g_j$ in \eqref{eq:cspineobjgeneral}, allowing for a wide range of sparsity structures. Incorporating more structured penalties is a promising direction for future work. For instance, one may consider hierarchical penalties that reflect the relationship between covariate effects on the mean and precision matrix, where a covariate influencing the precision matrix is expected to also affect the mean, but not necessarily vice versa.
Another direction is to develop tuning-free variants of the proposed method, for example by adapting ideas from the square-root lasso \citep{sqrtlasso}.
Additionally, methods for statistical inference could be developed by applying debiasing techniques to each nodewise regression \citep{meng_statistical_2025}. From a theoretical perspective, it would be of interest to establish minimax optimality of the estimator in \eqref{eq:cspineobj}. While our results achieve the same rate as in \citet{zhangli}, the optimality of this rate remains an open question.

\section*{Supporting Material}
The code and data supporting Section~\ref{sec:simulation} are available at \url{https://github.com/roobnloo/cspine-sim}.
The data that support the findings in Section~\ref{sec:realdata} are available from the corresponding author upon reasonable request.
The authors report there are no competing interests to declare.

\appendix

\section{Proofs of Theoretical Results} \label{appdx:proofs}
Our proof strategies for Theorems~\ref{thm:l2error} and \ref{thm:supportrecovery} follow that of \citet{zhangli}, with modifications made to accommodate our concatenated design matrix $\mUW$.
We will first recall several lemmas from the literature and then prove some supporting lemmas before moving onto the main proofs.

\subsection{Technical Lemmas}

\begin{lemma}[\citealt{bellec} Lemma 1]
\label{lem:1}
Let $g: \R^d \to \R$ be any convex function and let
\[
\hvec\beta\in \argmin_{\vec\beta\in\R^d}\sets*{\norm{\vec y - \mat H \vec\beta}_2^2 + g(\vec\beta)}
\]
where $\mat H \in \R^{n \times d}$ and $\vec y \in \R^n$. Then for all $\vec\beta\in \R^d$,
\[
\frac{1}{2n}\norm{\vec y - \mat H \hvec\beta}_2^2 + g(\hvec\beta) + \frac{1}{2n}\norm{\mat H(\hvec\beta - \vec\beta)}_2^2 \leq \frac{1}{2n}\norm{\vec y - \mat H \vec\beta}_2^2 + g(\vec\beta).
\]
\end{lemma}

\begin{lemma}[\citealt{graybill} Theorem F]
\label{lem:2}
Let $\vec \ep \sim N_p(\vec 0, \sigma^2 \mat I_p)$ and let $A$ be a $p\times p$ idempotent matrix with rank $r \leq p$.
Then $\vec \ep^\top A \vec \ep/\sigma^2 \sim \chi^2_r$.
\end{lemma}

\begin{lemma}[\citealt{laurent} Lemma 1]
\label{lem:3}
Suppose that $U\sim \chi^2_r$. For any $x>0$ it holds that
\[\PP(U - r\geq 2\sqrt{rx} + 2x)\leq e^{-x}.\]
\end{lemma}

\begin{lemma}[\citealt{vershynin} Proposition 5.16]
\label{lem:4}
Let $X_1,\dotsc, X_n$ be independent, mean zero sub-exponential random variables. Let $v_1 = \max_i \norm{X_i}_{\psi_1}$
where $\norm{\cdot}_{\psi_1}$ is the sub-exponential norm. Then there exists a constant $c$ such that for any $t>0$ we have
\[
\PP\left(\abs*{\sum_{i=1}^n X_i} \geq t\right) \leq 2\exp\sets*{-c \min\parens*{\frac{t^2}{v_1^2n}, \frac{t}{v_1}}}.
\]
\end{lemma}

Lemma~\ref{lem:kuchi} comes from Theorem~4.1 in \citet{kuchi} applied to marginally sub-exponential random vectors.
\begin{lemma}[\citealt{kuchi} Theorem 4.1]
\label{lem:kuchi}
Let $\bm Z_1,\dotsc, \bm Z_n$ be independent random vectors in $\R^p$ and let $\mat Z = [\bm Z_1, \dotsc, \bm Z_n]^\T$.
Assume each element of $\bm Z_i$ is sub-exponential with bounded sub-exponential norm for all $i\in [n]$;
$\norm{Z_{ij}}_{\psi_1} = \sup_{d\geq 1} d^{-1}(\EE\abs{Z_{ij}}^d)^{1/d} \leq K_z$ for all $i\in [n]$ and $j\in [p]$ for some constant $K_z > 0$.
Let $\hvec\Sigma_{\mat Z} = \mat Z^\top \mat Z/n$ and $\mat\Sigma_{\mat Z} = \EE[\mat Z^\top \mat Z/n]$.
Define
\[
\Upsilon_{n} = \max_{j,k} \frac{1}{n} \sum_{i=1}^n \Var(Z_{ij}Z_{ik}).
\]
Then we have
\[
\sup_{\norm{\vec v}_0 \leq k, \norm{\vec v}_2\leq 1}
\abs*{\vec v^\top (\hvec\Sigma_{\mat Z} - \mat\Sigma_{\mat Z}) \vec v} \precsim k \sqrt{\frac{\Upsilon_n \log p}{n}} + K_z^2\frac{k(\log n\log p)^2}{n}
\]
with probability at least $1 - O(1/p)$.
\end{lemma}
With Assumptions \ref{assump:bounded_u_element} and \ref{assump:subgaussian_x}, we may apply Lemma~\ref{lem:kuchi} to our design matrix $[\mat U, \mat W]$ because
it is elementwise sub-exponential, each entry being the product of a sub-Gaussian and a bounded random variable.
Furthermore, we have by the Cauchy-Schwarz inequality that
\[
\max_{\ell, k} \frac{1}{n} \sum_{i=1}^n \Var\parens*{[\mat U, \mat W]_{i\ell}[\mat U, \mat W]_{ik}}
\leq \max_{\ell_1, \ell_2, \ell_3, \ell_4} \EE\parens*{ X_{\ell_1}^{(1)^2} X_{\ell_2}^{(1)^2} U_{\ell_3}^{(1)^2} U_{\ell_4}^{(1)^2}} = O(1)
\]
since the entries of $\vec X^{(1)}$ and $\vec U^{(1)}$ have bounded moments.
Thus $\Upsilon_n = O(1)$ in our setting.

\begin{lemma}[\citealt{loh} Lemma 12]
\label{lem:loh}
Let $\mat\Sigma \in \R^{p\times p}$ be a symmetric matrix such that $\abs{\vec v^\T \mat\Sigma \vec v} \leq \delta_1$ for all $\vec v \in \R^p$ with $\norm{\vec v}_2 = 1$ and $\norm{\vec v}_0 \leq 2s$.
It holds for all $\vec v\in\R^p$ that
\[
\abs{\vec v^\T \mat\Sigma \vec v} \leq 27\delta_1\parens*{\norm{\vec v}_2^2 + \frac{1}{s}\norm{\vec v}_1^2}.
\]
\end{lemma}

Treating $j$ as fixed and suppressing the dependence of $\mat W_{-j}$ on $j$, recall that the sample and population Gram matrices are defined as
\[\hSUW = \frac{1}{n}[\mU , \mW]^\T [\mU , \mW],\quad \SUW = \EE\parens*{\frac{1}{n}[\mU , \mW]^\T [\mU , \mW]}.\]

\begin{lemma}
\label{lem:operatorboundUW}
For a set of indices $S\subset [p(q+1) - 1]$,
denote by $[\mU, \mW]_S$ the submatrix of $[\mU, \mW]$ with columns indexed by $S$.
Under Assumptions~\ref{assump:bounded_u_element}--\ref{assump:sparsity_scale},
there exist constants $M_{uw}$ and $C_0$ such that with probability at least $1 - C_0\exp(-\log(pq))$ we have
\[\frac{1}{n}\norm*{[\mU , \mW]_S}_{\mathrm{op}}^2 \leq M_{uw}\]
for all $S$ satisfying $\abs{S} \leq \hat s^{\mathrm{max}}_\gamma + \hat s^{\mathrm{max}}_\beta$, provided that $(\hat s_{\gamma}^\mathrm{max} + \hat s_\beta^\mathrm{max})\log(pq) = O(\sqrt{n})$, as assumed in Theorem~\ref{thm:l2error}.
\end{lemma}
\begin{proof}
Letting $k = \abs{S}$, it suffices to show that
$\sup_{\norm{\vec v}_0 \leq k,\, \norm{\vec v}_2\leq 1} \vec v^\T\hSUW\vec v$
is bounded with the provided probability. We may write
\begin{align*}
\sup_{\norm{\vec v}_0 \leq k,\, \norm{\vec v}_2\leq 1} \vec v^\T\hSUW\vec v
&= \sup_{\norm{\vec v}_0 \leq k,\, \norm{\vec v}_2\leq 1}
\sets*{\vec v^\T\parens*{\hSUW - \SUW}\vec v + \vec v^\T\SUW\vec v}.
\end{align*}
The second term is bounded by Assumption~\ref{assump:spectralboundSUW}.
For the first term, by Lemma~\ref{lem:kuchi} we have
\[
\sup_{\norm{\vec v}_0 \leq k,\, \norm{\vec v}_2\leq 1}
\sets*{\vec v^\T\parens*{\hSUW - \SUW}\vec v}
\precsim k\parens*{\frac{\log(pq)}{n}}^{1/2} + \frac{k\log(pq)^2}{n/(\log n)^2}
\]
with probability at least $1 - C_0\exp(-\log(pq))$
and we see that the right-hand side is $o(1)$ by Assumption~\ref{assump:sparsity_scale}.
\end{proof}

\subsection{Proof of Theorem~\ref{thm:l2error}}
For ease of notation, we will drop the dependence of $\vec\gamma_j$, $\vec\beta_j$, $\vec\ep_j$ and $\mat W_{-j}$ on $j$.
Let $S_\beta, S_\gamma, \hat S_\beta, \hat S_\gamma$ be the support sets of $\vec\beta, \vec\gamma, \hvec\beta, \hvec\gamma$, respectively.
Let $S_{\beta,g} = \sets{h\in[q]: \vec\beta_h \neq \vec 0}$ index the blocks $\vec\beta_h$ of $\vec \beta$ that are not identically zero and let $\hat S_{\beta, g}$ be the corresponding block indices for $\hvec\beta$.
For any vector $\vec v$ and set of block indices $S$, let $\vec v_{(S)}$ denote the sub-vector containing blocks in $S$.
Let $s_\beta, s_\gamma, s_{\beta, g}, \hat s_\beta, \hat s_\gamma, \hat s_{\beta, g}$ be the number of elements in
$S_\beta, S_\gamma, S_{\beta, g}, \hat S_\beta, \hat S_\gamma, \hat S_{\beta, g}$, respectively.
Our proof occurs in three steps.

\textit{Step 1:}
In this step we bound the error $\norm{\mU\vec\nu + \mW\vec\Delta}_2^2/n$ by the stochastic term $\<\vec\ep, \mU\vec\nu + \mW\vec\Delta\>/n$, which is then bounded by a projection of $\vec\ep$ onto the columns of $\mUW$.

Since our penalty function
\[
g(\vec\gamma, \vec\beta) = \lambda\norm{\vec\gamma}_1 + \lambda\norm{\vec\beta}_1 + \lambda_g\norm{\vec\beta_{-0}}_{1,2}
\]
is convex, by Lemma~\ref{lem:1} we have
\[
\frac{1}{2n}\norm{\vec x - \mU\hvec\gamma - \mat W\hvec\beta}_2^2 + g(\hat{\vec\gamma}, \hat{\vec\beta})
+ \frac{1}{2n} \norm*{\mU\vec\nu + \mat W\vec\Delta}_2^2
\leq
\frac{1}{2n}\norm{\vec x - \mU\vec\gamma - \mat W\vec\beta}_2^2 + g({\vec\gamma}, {\vec\beta})
\]
where $\vec\nu = \hat{\vec\gamma} - \vec\gamma$ and $\vec\Delta = \hat{\vec\beta} - \vec\beta$.
Since $\vec\ep = \vec x - \mU\vec\gamma - \mat W\vec\beta$ we may write
\begin{align*}
\frac{1}{2n}\norm{\vec x - \mU\hvec\gamma - \mat W\hvec\beta}_2^2 &= \frac{1}{2n}\norm{\vec\ep - \mU\vec\nu - \mat W\vec\Delta}_2^2 \\
&= \frac{1}{2n}\norm{\vec\ep}_2^2 - \frac{1}{n}\<\vec\ep, \mU\vec\nu + \mat W\vec\Delta\> + \frac{1}{2n}\norm*{\mU\vec\nu + \mat W\vec\Delta}_2^2.
\end{align*}
Plugging this into the previous expression and substituting the penalty expression then yields
\begin{align*}
&\frac{1}{n} \norm{\mU\vec\nu + \mat W\vec\Delta}_2^2 + \lambda\norm{\hvec\gamma}_1 + \lambda\norm{\hvec\beta}_1 + \lambda_g\norm{\hvec\beta_{-0}}_{1,2} \\
&\qquad\leq \frac{1}{n} \<\vec\ep, \mU\vec\nu + \mat W\vec\Delta\> + \lambda\norm{\vec\gamma}_1 + \lambda\norm{\vec\beta}_1 + \lambda_g\norm{\vec\beta_{-0}}_{1,2}.
\end{align*}
Using the triangle inequality for the $\ell_1$ norm (and analogously for the group norm) yields
\begin{equation}
\label{eqn:s2}
\begin{split}
&\frac{1}{n} \norm{\mU\vec\nu + \mat W\vec\Delta}_2^2 + \lambda\norm{\vec\nu_{S_\gamma^c}}_1 + \lambda\norm{\vec\Delta_{S_\beta^c}}_1 + \lambda_g\norm{\vec\Delta_{(S_{\beta, g}^c)}}_{1,2} \\
&\qquad\leq \frac{1}{n} \<\vec\ep, \mU\vec\nu + \mat W\vec\Delta\> + \lambda\norm{\vec\nu_{S_\gamma}}_1 + \lambda\norm{\vec\Delta_{S_\beta}}_1 + \lambda_g\norm{\vec\Delta_{(S_{\beta, g})}}_{1,2}.
\end{split}
\end{equation}

Now let $\cI$ and $\cJ$ be arbitrary index sets of the columns of $\mat U$ and $\mat W$ respectively.
Denote by $\cP_{\cI, \cJ}$ the orthogonal projection onto the columns of $\mUW$ indexed by $(\cI, \cJ)$.
Let $\cI_0 = S_\gamma \cup \hat S_\gamma$ and $\cJ_0 = S_\beta \cup \hat S_\beta$ denote the unions of the true and estimated support sets of $\vec\gamma$ and $\vec\beta$.
We seek to bound the stochastic term
\begin{equation}
\label{eqn:s3}
\begin{split}
\<\vec\ep, \mU\vec\nu + \mat W\vec\Delta\> &= \<\cP_{\cI_0, \cJ_0}(\vec\ep), \mU\vec\nu + \mat W\vec\Delta\> \\
& \leq \norm{\cP_{\cI_0, \cJ_0}(\vec\ep)}_2 \norm{\mU\vec\nu + \mat W\vec\Delta}_2 \\
& \leq \frac{1}{2a_1}\norm{\mU\vec\nu + \mat W\vec\Delta}_2^2 + \frac{a_1}{2} \norm{\cP_{\cI_0, \cJ_0}(\vec\ep)}_2^2.
\end{split}
\end{equation}

\textit{Step 2:}
Following \citet{zhangli}, we first bound the term $\norm{\cP_{\cI_0, \cJ_0}(\vec\ep)}_2^2$ with a counting argument.
For fixed $s'_\gamma$, $s'_\beta$, and $s'_{\beta, g}$, let
\[\cH(s'_\gamma, s'_\beta, s'_{\beta, g}) = \sets{(\cI, \cJ) \subset [q] \times [(p-1)(q+1)]: \abs{\cI} = s'_\gamma,\, \abs{\cJ} = s'_\beta,\, \abs{g(\cJ)} = s'_{\beta, g}}\]
where $g(\cJ)$ is the number of nonzero groups of $\vec\beta_\cJ$.
We show that
\[
\log\abs{\cH} \leq s'_\gamma \log\frac{eq}{s'_\gamma} + s'_{\beta, g}\log\frac{eq}{s'_{\beta,g}} + s'_\beta \log(ep).
\]
Define $k_0$ to be the exponential of the right-hand side, so that $\abs{\cH} \leq k_0$.
For any $(\cI, \cJ)\in \cH$, Lemmas~\ref{lem:2} and \ref{lem:3} yield, after a union bound,
\begin{equation}
\label{eq:projbound0}
\PP\parens*{\sup_{(\cI, \cJ)\in\cH}\norm{\cP_{\cI, \cJ}(\vec\ep)}_2^2 \geq 5\sigma^2_\ep {\bracks*{s'_\gamma \log\frac{eq}{s'_\gamma} + s'_{\beta, g}\log\frac{eq}{s'_{\beta,g}} + s'_\beta \log(ep)}} + \sigma^2_\ep t} \leq e^{-t/3}.
\end{equation}
Bounding $\norm{\cP_{\cI_0, \cJ_0}(\vec\ep)}_2^2$ and $r = \sup r(s'_\gamma, s'_\beta, s'_{\beta,g})$ via the union bound and the KKT conditions gives the concentration bound
\begin{equation}
\label{eqn:rbound}
\PP\sets*{r \geq \tilde M \sigma^2_\ep(s_\gamma \log(eq/s_\gamma) + s_\beta\log(e p) + s_{\beta,g}\log(eq/s_{\beta,g}))}
\leq c_1\exp\sets*{-c_2 E_j/\sigma^2_\ep}.
\end{equation}
Using the KKT optimality conditions and Lemma~\ref{lem:operatorboundUW},
\begin{equation}
\label{eqn:projbound2}
\lambda^2 \hat{s}_\gamma + \lambda^2\hat s_\beta + \lambda_g^2 \hat s_{\beta, g}
\leq
\frac{2}{n}M_{uw} \norm*{\cP_{\cI_0, \cJ_0}(\vec\ep)}_2^2 + \frac{2}{n}M_{uw} \norm*{\mU{\vec\nu} + \mat W{\vec\Delta}}_2^2.
\end{equation}
Combining \eqref{eq:projbound0}--\eqref{eqn:projbound2} and plugging into \eqref{eqn:s2}--\eqref{eqn:s3} gives
\begin{equation}
\label{eqn:s11}
\begin{split}
&\frac{1}{n}\norm{\mU\vec\nu + \mat W\vec\Delta}_2^2 + \lambda\norm{\vec\nu_{S^c_\gamma}}_1 + \lambda\norm{\vec\Delta_{S^c_\beta}}_1 + \lambda_g\norm{\vec\Delta_{S^c_{\beta,g}}}_{1,2} \\
&\quad\leq \frac{5a_1a_2}{2(a_2 -2)} E_j + \frac{a_1a_2}{2(a_2-2)}\cdot\frac{r}{n} + \lambda\norm{\vec\nu_{S_\gamma}}_1 + \lambda\norm{\vec\Delta_{S_\beta}}_1 + \lambda_g\norm{\vec\Delta_{S_{\beta,g}}}_{1,2}
\end{split}
\end{equation}
where $E_j = (\sigma^2_\ep/n)\cdot \parens*{s_\gamma\log(eq/s_\gamma) + s_\beta\log(ep) + s_{\beta,g}\log(eq/s_{\beta, g})}$.

\textit{Step 3:}
We proceed to bound the penalty term by $\norm*{\SUW^{1/2} \begin{pmatrix}\vec\nu \\ \vec\Delta\end{pmatrix}}$.
Standard $\ell_1, \ell_2$ norm inequalities along with Assumption~\ref{assump:spectralboundSUW} lead to
\begin{align*}
\lambda\norm{\vec\nu_{S_\gamma}}_1 + \lambda\norm{\vec\Delta_{S_\beta}}_1 + \lambda_g\norm{\vec\Delta_{S_{\beta,g}}}_{1,2}
&\leq \frac{2C}{\sqrt{\phi_0^\ast}}\sqrt{E_j}\norm*{\SUW^{1/2} \begin{pmatrix}\vec\nu \\ \vec\Delta\end{pmatrix} }_2 \\
&\leq a_3\frac{C^2}{\phi_0^\ast}E_j + \frac{1}{a_3}\norm*{\SUW^{1/2} \begin{pmatrix}\vec\nu \\ \vec\Delta\end{pmatrix} }_2^2
\end{align*}
for any $a_3 > 0$.
We then bound the difference between $\norm*{\mU\vec\nu + \mat W\vec\Delta}_2^2/n$ and $\norm*{\SUW^{1/2} \begin{pmatrix}\vec\nu \\ \vec\Delta\end{pmatrix} }_2^2$ using Lemmas~\ref{lem:kuchi} and \ref{lem:loh}.
Choosing $a_1=2$, $a_2=6$, $a_3=6$ and applying Assumption~\ref{assump:spectralboundSUW}, we obtain
\begin{equation*}
\norm{\vec\nu}_2^2 + \norm{\vec\Delta}_2^2 \precsim
\frac{\sigma^2_\ep}{n} \parens*{s_\gamma\log(eq/s_\gamma) + s_\beta\log(ep) + s_{\beta,g}\log(eq/s_{\beta, g})}
\end{equation*}
with probability at least $1 - C_1\exp\sets*{-C_2(s_\gamma \log(eq/s_\gamma) + s_\beta\log(e p) + s_{\beta,g}\log(eq/s_{\beta,g}))}$
for some positive constants $C_1, C_2$. $\blacksquare$

\subsection{Proof of Theorem~\ref{thm:supportrecovery}}
The proof occurs in three steps.
\textit{Step 1:}
We show that with high probability
\begin{equation}
\label{eq:step1}
    \norm*{\hSUW \begin{pmatrix}\vec\nu \\ \vec\Delta\end{pmatrix}}_\infty \leq \frac{3}{2}(\lambda + \lambda_g).
\end{equation}
By the KKT conditions we know that an optimizer $(\hvec\gamma, \hvec\beta)$ satisfies
\[
\abs*{\parens*{[\mU, \mat W]^\T (\vec x_j - \mU\hvec\gamma - \mat W\hvec\beta)/n}_{\ell}} \leq \lambda + \lambda_g
\]
for all $\ell\in [p(q+1)-1]$.
By the triangle inequality, to show \eqref{eq:step1} it suffices to show that with high probability
\begin{equation}
\label{eqn:step1-suff}
\norm*{\frac{1}{n} [\mU, \mat W]^\T\vec\ep}_\infty \leq \frac{\lambda + \lambda_g}{2}.
\end{equation}
The products $U_\ell\ep$ and $X_\ell U_h\ep$ are sub-exponential; by Lemma~\ref{lem:4} and a union bound with $\log p \asymp \log q$,
\[
\PP\parens*{\frac{1}{n} \norm*{[\mU, \mat W]^\T \vec\ep}_\infty \geq \frac{\lambda + \lambda_g}{2}}
\leq 2\exp(-c_3\log p)
\]
for some $c_3 > 0$.

We also establish the cone condition
\begin{equation}
\label{eq:step1cone}
\norm{\vec\nu_{S_\gamma^c}}_1 + \norm{\vec\Delta_{S_\beta^c}}_1 \leq
 4\tau_j(\norm{\vec\nu_{S_\gamma}}_1 + \norm{\vec \Delta_{S_\beta}}_1)
\end{equation}
using the optimality of $(\hvec\gamma, \hvec\beta)$ and H\"{o}lder's inequality on the stochastic term, where $\tau_j = 1 + \sqrt{(s_\gamma + s_\beta)/s_{\beta,g}}$ as in Assumption~\ref{assump:coherence}.

\textit{Step 2:}
We bound the diagonal and off-diagonal entries of $\hSUW$.
By Lemma~\ref{lem:4} and Assumption~\ref{assump:spectralboundSUW},
\begin{equation}
\PP\parens*{\frac{\phi^\ast_0}{2} \leq \hSUW(\ell, \ell) \leq 2\phi^\ast_1} \geq 1- 2\exp(-c_6 n).
\end{equation}
By Assumption~\ref{assump:coherence},
\begin{equation}
\label{eq:s26}
\PP\parens*{ -\frac{1}{c_0(1+8\tau_j)(s_\beta+s_\gamma)}\leq \hSUW(k,\ell)\leq \frac{3}{c_0(1+8\tau_j)(s_\beta+s_\gamma)}}
\geq 1 - 2\exp(-c_7n).
\end{equation}

\textit{Step 3 and Final Step:}
Define $\tilde S = S_\gamma \cup \sets{q + i\mid i\in S_\beta}$.
Using the bounds from Steps~1--2, we establish the restricted eigenvalue condition
\begin{equation}
\label{eq:step3}
\frac{1}{n}\norm*{[\mU,\mat W] \begin{pmatrix} \vec\nu \\ \vec\Delta \end{pmatrix}}_2^2
\geq \parens*{\frac{\phi^\ast_0}{2} - \frac{1}{c_0}}\norm*{\begin{pmatrix} \vec\nu \\ \vec\Delta \end{pmatrix}_{\tilde S}}_2^2.
\end{equation}
Combining \eqref{eq:step1}, \eqref{eq:step1cone}, and \eqref{eq:step3} yields
\[
\norm*{\begin{pmatrix} \vec\nu \\ \vec\Delta \end{pmatrix}_{\tilde S}}_2
\leq
3(\lambda+\lambda_g) (1+4\tau_j) \parens*{\frac{c_0}{c_0\phi^\ast_0 - 2}} \sqrt{s_\gamma+s_\beta}.
\]
Finally, applying the $\ell_\infty$ bound from \eqref{eq:step1} and the matrix inequality
\[
\norm*{\begin{pmatrix} \vec\nu \\ \vec\Delta \end{pmatrix}}_\infty \leq \frac{2}{\phi^\ast_0} \norm*{\hSUW \begin{pmatrix} \vec\nu \\ \vec\Delta \end{pmatrix}}_\infty +
\frac{6}{c_0\phi^\ast_0(1+8\tau_j)(s_\gamma+s_\beta)}\norm*{\begin{pmatrix} \vec\nu \\ \vec\Delta \end{pmatrix}}_1,
\]
we obtain
\[
\norm*{\begin{pmatrix} \vec\nu \\ \vec\Delta \end{pmatrix}}_\infty
=
\frac{3}{\phi^\ast_0}(\lambda+\lambda_g)\parens*{1 + \frac{6(1+4\tau_j)^2}{(1+8\tau_j)(c_0\phi^\ast_0 - 2)}}
\]
as desired. $\blacksquare$

\section{Sufficient Conditions for Theoretical Results}
\label{sec:sufficientconditions}
While Assumptions~\ref{assump:bounded_u_element}--\ref{assump:sparsity_scale} are enough to prove Theorem~\ref{thm:l2error} and the addition of Assumption~\ref{assump:coherence} for Theorem~\ref{thm:supportrecovery}, it is desirable to present a set of interpretable sufficient conditions.
The key difference in the theoretical analysis of our proposed method versus that of \citet{zhangli} is that, due to the stage 1 centering in \eqref{eq:zhanglistage1}, the columns of the design matrix in \citet{zhangli} are centered. By contrast, the design matrix $[\mat U, \mat W]$ of our proposed method contains the terms $x^{(i)}_j u^{(i)}_h$ where it does not necessarily hold that $\EE(x^{(i)}_j \given \vec u^{(i)}) = 0$, but does obviate the need to bound a noisy design matrix.

One consequence of the above is that the matrix $\EE(\mat W_{\mathrm{ZL}}^\top \mat W_{\mathrm{ZL}}/n)$ is a covariance matrix whose spectrum is easily bounded via bounds on the eigenvalues of $\Cov(\bm U)$ and $\Cov(\bm X\given \bm U)$; these are Assumptions~1 and 2 in \citet{zhangli}.
Meanwhile, our analogous matrix $\SUW$ is an (uncentered) Gram matrix. Nevertheless, Theorems~\ref{thm:l2error} and \ref{thm:supportrecovery} follow from interpretable bounds on the eigenvalues of $\Cov(\bm U)$ and $\Cov(\bm X\given \bm U)$, plus additional assumptions that control the mean term.

\subsection{Sufficient Conditions for Theorem~\ref{thm:l2error}}
\begin{assumption}
\label{assump:covu}
The covariates $\sets{\vec u^{(i)}}_{i=1}^n$ are i.i.d. mean zero random vectors with covariance matrix satisfying
\[\phi_0 \leq \lammin(\Cov(\vec u^{(i)})) \leq \lammax(\Cov(\vec u^{(i)})) \leq \phi_1\]
for some constants $0 < \phi_0 \leq \phi_1 < \infty$.
\end{assumption}
\begin{assumption}
\label{assump:covx}
Recall that $\mat\Sigma(\vec u) = \Cov(\bm X\given \bm U = \vec u)$.
Suppose for constants $\psi_0, \psi_1 > 0$ we have
\[\psi_0 \leq \lammin(\mat\Sigma(\vec u^{(i)})) \leq \lammax(\mat\Sigma(\vec u^{(i)})) \leq \psi_1.\]
\end{assumption}
Assumptions~\ref{assump:covu} and \ref{assump:covx} are analogous to Assumptions~1 and 2 in \citet{zhangli} respectively, and are standard for results in the high-dimensional setting.

Next, as our design matrix is the concatenated $[\mat U, \mat W]$, we must restrict the cross-covariance between $\mat U$ and $\mat W$.
\begin{assumption}
\label{assump:cross_covar}
For any vectors $\vec a\in \R^{q}$, $\vec b \in \R^{(p-1)(q+1)}$ define the correlation of projections
\[
\rho^{(i)}(\vec a, \vec b) = \operatorname{Corr}\parens*{\vec a^\top \vec u^{(i)}, \vec b^\top \vec w^{(i)}}
\]
and let $\rho^{(i)}_{\min} = \min_{\norm{\vec a}_2 = 1, \norm{\vec b}_2 = 1} \rho^{(i)}(\vec a, \vec b)$ denote the minimal canonical correlation between $\vec u^{(i)}$ and $\vec w^{(i)}$.
We assume there exists a constant $\rho_0$ satisfying $0 \leq \rho_0 <1$ such that $\rho^{(i)}_{\min} > -\rho_0$ for all $i\in[n]$.
\end{assumption}

\begin{assumption}
\label{assump:bounded_u_mu_expect}
There exist constants $M_\mu, M_{u,1} > 0$ such that $\EE\norm{\vec\mu(\vec u)}_2^4 \leq M_\mu$ and $\EE \norm{\vec u}_2^4 \leq M_{u, 1}$.
\end{assumption}
\begin{assumption}
\label{assump:bounded_mean_infty}
There exists a constant $M_\infty > 0$ such that $\norm{\vec\mu(\vec u)}_\infty \leq M_\infty$.
\end{assumption}

Together, Assumptions~\ref{assump:covu} through \ref{assump:bounded_mean_infty} are sufficient to imply Assumptions~\ref{assump:spectralboundSUW} and \ref{assump:subgaussian_x}, as summarized in the following propositions.
\begin{proposition}
\label{prop:spectralboundUW_sufficient}
Assumptions~\ref{assump:bounded_u_element} and \ref{assump:covu}--\ref{assump:bounded_u_mu_expect} imply Assumption~\ref{assump:spectralboundSUW}.
\end{proposition}
\begin{proposition}
\label{prop:subgaussian_x_sufficient}
Assumptions~\ref{assump:bounded_u_element}, \ref{assump:covx}, and \ref{assump:bounded_mean_infty} imply Assumption~\ref{assump:subgaussian_x}.
\end{proposition}

\subsection{Sufficient Conditions for Theorem~\ref{thm:supportrecovery}}
With Assumption~\ref{assump:bounded_mean_infty} we have
\begin{align*}
  \max_{k\neq \ell} \abs{\SUW(k, \ell)}
  \leq \max_{\ell_1 \neq \ell_2}\abs*{\bracks*{\Cov(\vec u^{(1)})}_{\ell_1, \ell_2}} \times \abs*{\max_{\ell_3 \neq \ell_4}\sup_{\vec u} \bracks*{\Cov(\vec x^{(1)}\given \vec u)}_{\ell_3, \ell_4} + M_\infty^2}.
\end{align*}
Therefore, Assumption~\ref{assump:coherence} is satisfied when the covariances of $\bm X$ and $\bm U$ are small and the mean of $\bm X$ is small.

\subsection{Proof of Proposition~\ref{prop:spectralboundUW_sufficient}}
We wish to prove that Assumptions \ref{assump:bounded_u_element} and \ref{assump:covu}--\ref{assump:bounded_u_mu_expect} imply Assumption~\ref{assump:spectralboundSUW}.
It suffices to upper and lower bound $\vec v^\top \SUW \vec v$ for all unit vectors $\vec v$.
Express the block matrix form as
\[
\SUW =
\begin{pmatrix}
    \mat\Sigma_{\mat U \mat U} & \SUW^{(12)} \\
    \SUW^{(12)\top} & \mat\Sigma_{\mat W \mat W}
\end{pmatrix}
\]
and let $\vec w = (\vec x, u_1\vec x, \dotsc, u_q\vec x) \in \R^{(p-1)(q+1)}$ be a row of the data matrix $\mat W$.
Express an arbitrary unit vector $\vec v$ as $(\vec a, \vec b)$ where $\vec a \in \R^q$ and $\vec b \in \R^{(p-1)(q+1)}$.
Further express $\vec b$ as $(\vec b_0, \vec b_1, \dotsc, \vec b_q) $ where $\vec b_h \in \R^{p-1}$, $h = 0, \dotsc, q$.
Let $\mat B = [\vec b_1, \dotsc, \vec b_q] \in \R^{(p-1)\times q}$ and define $c(\vec u) = \vec b_0 + \mat B \vec u$.

\textit{Lower bound:}
Since $\vec u$ is centered, $\vec a^\top \mat\Sigma_{\mat U \mat U} \vec a \geq \phi_0 \norm{\vec a}_2^2$ by Assumption~\ref{assump:covu}.
By Assumptions~\ref{assump:covu} and \ref{assump:covx} and the law of total variance,
\[
\vec b^\T \mat\Sigma_{\mat W \mat W} \vec b \geq \psi_0\norm{\vec b_0}_2^2 + \psi_0\phi_0 \norm{\mat B}_F^2
\geq \phi_0'\norm{\vec b}_2^2
\]
where $\phi_0' = \min(\psi_0, \psi_0\phi_0)$.
By Assumption~\ref{assump:cross_covar},
\[
\vec a^\top \SUW^{(12)}\vec b \geq -\rho_0\parens*{\vec a^\top \mat\Sigma_{\mat U \mat U} \vec a}^{1/2}\parens*{\vec b^\top \mat\Sigma_{\mat W \mat W} \vec b}^{1/2}.
\]
Completing the square yields
\[
\vec v^\T \SUW \vec v \geq (1-\rho_0^2)\frac{\phi_0\phi_0'}{\phi_0 + \phi_0'} =: \phi_0^\ast.
\]

\textit{Upper bound:}
By Assumption~\ref{assump:covu}, $\vec a^\top \mat\Sigma_{\mat U \mat U} \vec a \leq \phi_1 \norm{\vec a}_2^2$.
By Assumptions~\ref{assump:covu}, \ref{assump:covx}, and \ref{assump:bounded_u_mu_expect},
\[
\vec b^\T \mat\Sigma_{\mat W \mat W} \vec b \leq (\phi_1' + \phi_2')\norm{\vec b}_2^2
\]
where $\phi_1' = \max(\psi_1, \psi_1\phi_1)$ and $\phi_2' = \sqrt{8}M_\mu^{1/2}(1+M_{u,1}^{1/2})$.
By the AM-GM inequality and all in all,
\[
\vec v^\T \SUW \vec v \leq 2\max(\phi_1, \phi_1' + \phi_2') =: \phi^\ast_1. \quad\blacksquare
\]

\subsection{Proof of Proposition~\ref{prop:subgaussian_x_sufficient}}
We show that Assumptions~\ref{assump:bounded_u_element}, \ref{assump:covx}, and \ref{assump:bounded_mean_infty} imply Assumption~\ref{assump:subgaussian_x}.
Write $X_i - \EE X_i = (X_i - \mu_i(\vec u)) + (\mu_i(\vec u) - \EE[\mu_i(\vec u)])$.
Conditionally on $\vec u$, $X_i - \mu_i(\vec u) \sim \mathcal N(0, [\mat\Sigma(\vec u)]_{ii})$, which is sub-Gaussian with norm bounded by $K_1\sqrt{\psi_1}$ by Assumption~\ref{assump:covx}.
By Assumption~\ref{assump:bounded_mean_infty}, $\abs{\mu_i(\vec u) - \EE[\mu_i(\vec u)]} \leq 2 M_\infty$, so $\mu_i(\vec u) - \EE[\mu_i(\vec u)]$ is also sub-Gaussian with bounded norm. $\blacksquare$

\bibliography{biblio}

\end{document}